\documentclass
[
    twoside,                 
    openright,               
    cleardoublepage = empty, 
    fontsize = 12 pt,        
    american,                
    captions = tableheading, 
    numbers = noenddot,      
    footheight = 35 pt,      
]
{scrbook}

\newif\ifprintVersion   
\newif\ifprofessionalPrint 
\newif\iffancyTheorems  
\newif\ifboldNumberSets 
\newif\ifbachelorThesis 

\printVersionfalse
\professionalPrintfalse
\fancyTheoremstrue
\boldNumberSetstrue
\bachelorThesistrue

\newcommand*{\printTitle}{}
\newcommand*{\printGermanTitle}{}
\newcommand*{\myTitle}[2]{\renewcommand*{\printTitle}{#1}\renewcommand*{\printGermanTitle}{#2}}
\newcommand*{\printTitleBold}{\textbf{\printTitle}}

\newcommand*{\printAuthor}{}
\newcommand*{\myName}[1]{\renewcommand*{\printAuthor}{#1}}

\newcommand*{\printProgram}{}
\newcommand*{\myProgram}[1]{\renewcommand*{\printProgram}{#1}}

\newcommand*{\printDateReceived}{}
\newcommand*{\dateOfHandingIn}[1]{\renewcommand*{\printDateReceived}{#1}}

\newcommand*{\printSubject}{}
\newcommand*{\mySubject}[1]{\renewcommand*{\printSubject}{#1}}

\newcommand*{\printKeywords}{}
\newcommand*{\myKeywords}[1]{\renewcommand*{\printKeywords}{#1}}

\newcommand*{\printNameOfSupervisor}{}
\newcommand*{\nameOfMySupervisor}[1]{\renewcommand*{\printNameOfSupervisor}{#1}}

\newcommand*{\printAdditionalExaminers}{}
\newcommand*{\additionalExaminers}[1]{\renewcommand*{\printAdditionalExaminers}{#1}}

\newlength{\extraborderlength}
\newcommand*{\extraBorder}[1]{\setlength{\extraborderlength}{#1}}

\newlength{\mybindingcorrection}
\newcommand*{\bindingCorrection}[1]{\setlength{\mybindingcorrection}{#1}} 

    \extraBorder{3 mm}

    \bindingCorrection{6 mm}

    \myTitle{Stability in Single-Peaked Strategic Resource Selection Games}{Stabilität bei der Strategischen Auswahl von Ressourcen mit Single-Peaked Präferenzen}

    \myName{Henri Zeiler}

    \myProgram{IT Systems Engineering}

    \dateOfHandingIn{01. Oktober 2024}

    \nameOfMySupervisor{Prof.\,Dr. Tobias Friedrich}

    \additionalExaminers{Dr. Pascal Lenzner\newline Simon Krogmann}

    \mySubject{A cool bachelor/master thesis.}

    \myKeywords{bachelor--master thesis | world-changing | very important | please like and share and subscribe}

\usepackage[utf8]{inputenc} 
\usepackage[T1]{fontenc}    
\usepackage
[
    ngerman,         
    main = american, 
]
{babel}                     

\input glyphtounicode
\usepackage{calc} 

\newlength{\myparindent}
\newlength{\myparskip}
\setfootnoterule{0 cm}

\deffootnote[1.2 em]{1.2 em}{0 em}{\makebox[1.4 em][l]{\textbf{\thefootnotemark}}}

\makeatletter%
    \@removefromreset{footnote}{chapter}%
\makeatother

\usepackage[dvipsnames]{xcolor} 

\definecolor{stroke1}{HTML}{2574A9} 

\colorlet{captionlabel}{black}
\colorlet{footerpagenr}{black}
\colorlet{footerchapter}{stroke1}
\colorlet{footerchaptername}{black}
\colorlet{footersection}{stroke1}
\colorlet{footersectionname}{black}
\colorlet{chapternumber}{stroke1}

\newlength{\mypaperwidth}
\newlength{\mypaperheight}
\newlength{\mybodywidth}
\newlength{\mybodyheight}
\newlength{\myoutermargin}
\ifprintVersion
    \ifprofessionalPrint
    \else
    \fi
\else
\fi

\newlength{\mytopmargin}
\ifprintVersion
    \ifprofessionalPrint
    \fi
\fi

\newlength{\myinnermargin}
\ifprintVersion
    \ifprofessionalPrint
    \fi
\fi

\newlength{\mybottommargin}
\ifprintVersion
    \ifprofessionalPrint
    \fi
\fi

\newcommand{\goldenratio}{1.618}

\newlength{\myheadsep} 
\ifprintVersion
    \ifprofessionalPrint
    \fi
\fi

\newlength{\myfootskip} 
\ifprintVersion
    \ifprofessionalPrint
    \fi
\fi

\newlength{\mymargininnersep} 
\newlength{\mymarginoutersep} 
\ifprintVersion
    \ifprofessionalPrint
    \fi
\fi

\newlength{\mymarginwidth} 
\newlength{\mymarginwidthwithinnersep} 
\usepackage
[
    \ifprintVersion
        \ifprofessionalPrint
            paperwidth = \mypaperwidth + 2\extraborderlength + \mybindingcorrection,
            paperheight = \mypaperheight + 2\extraborderlength,
        \else
            paperwidth = \mypaperwidth,
            paperheight = \mypaperheight,
        \fi
    \else
        paperwidth = \mypaperwidth,
        paperheight = \mypaperheight,
    \fi
    textwidth = \mybodywidth,
    textheight = \mybodyheight,
    outer = \myoutermargin,
    top = \mytopmargin,
    headsep = \myheadsep,
    footskip = \myfootskip,
    marginparsep = \mymargininnersep,
    marginparwidth = \mymarginwidth,
]
{geometry} 

\usepackage
[
]
{scrlayer-scrpage} 

\clearpairofpagestyles

\KOMAoptions
{%
    headwidth = \textwidth + \mymarginwidthwithinnersep,%
    footwidth = \myoutermargin : \textwidth,%
}

\automark[chapter]{chapter}
\automark*[section]{}

\lehead%
{%
    \begin{minipage}[b]{\mymarginwidth}%
        \small\raggedleft\normalfont\textsf{\textbf{\color{footerchapter}\chaptername\ \thechapter}}
    \end{minipage}
}
\cehead{\hspace*{\mymarginwidthwithinnersep}\parbox{\textwidth}{\raggedright\leftmark}}

\rohead%
{%
    \Ifstr{\rightmark}{\leftmark}%
    {%
        \begin{minipage}[b]{\mymarginwidth}%
            \small\raggedright\normalfont\textsf{\textbf{\color{footersection}Chapter\ \thechapter}}%
        \end{minipage}%
    }%
    {%
        \begin{minipage}[b]{\mymarginwidth}%
            \small\raggedright\normalfont\textsf{\textbf{\color{footersection}Section\ \thesection}}%
        \end{minipage}%
    }%
}
\cohead{\hspace*{-\mymarginwidthwithinnersep}\parbox{\textwidth}{\raggedleft\rightmark}}

\lefoot*%
{%
    \vspace*{1 ex}%
    {\color{stroke1}\rule{\myoutermargin - \mymargininnersep}{0.5 mm}}\\
    \begin{minipage}[b]{\myoutermargin - \mymargininnersep}%
        \raggedleft\normalfont\color{footerpagenr}\textbf{\thepage}%
    \end{minipage}%
}
\rofoot*%
{%
    {\color{stroke1}\rule{\myoutermargin - \mymargininnersep}{0.5 mm}}\\
    \begin{minipage}[b]{\myoutermargin - \mymargininnersep}%
        \raggedright\normalfont\color{footerpagenr}\textbf{\thepage}%
    \end{minipage}%
}

\usepackage{caption}
\usepackage{tikz} 

\newlength{\mytmpa}
\newlength{\mytmpb}

\renewcommand*{\partlineswithprefixformat}[3]%
{%
    #2
    \thispagestyle{empty}
    \setlength{\mytmpa}{0.618\mypaperwidth}%
    \setlength{\mytmpb}{0.382\mypaperheight}%
    \ifprintVersion
        \ifprofessionalPrint
            \setlength{\mytmpa}{0.618\mypaperwidth + \mybindingcorrection + \extraborderlength}%
            \setlength{\mytmpb}{0.382\mypaperheight + \extraborderlength}%
        \fi
    \fi
    \begin{tikzpicture}[overlay, remember picture]%
        \node [inner sep = 0, outer sep = 0, anchor = north] at (current page.north west)%
        {%
            \begin{tikzpicture}[overlay, remember picture]%
            \draw[color = stroke1, line width = 0.7 mm] (\mytmpa, 0) -- (\mytmpa, -\mytmpb);%
            \end{tikzpicture}%
        };%
        \node (align) [align = right, below = \mytmpb - 2 ex, inner sep = 0, outer sep = 0, anchor = north west] at (current page.north west)%
        {%
            \hspace{\mytmpa}\hspace{0.5 em}\partname\ \thepart\\[1 ex]
            \color{stroke1}#3%
        };%
    \end{tikzpicture}%
}
\RedeclareSectionCommand%
[%
    font = \normalfont\Huge\sffamily,
    prefixfont = \normalfont\Huge\sffamily,
]
{part}

\usepackage{etoolbox}

\newbool{chapterHasANumber}
\newbool{chapterHasAStar}
\renewcommand*{\chapterlinesformat}[3]%
{%
    \Ifnumbered{#1}{\setbool{chapterHasANumber}{true}}{\setbool{chapterHasANumber}{false}}%
    \Ifstr{#2}{}{\setbool{chapterHasAStar}{true}}{\setbool{chapterHasAStar}{false}}%
    \ifboolexpr{bool{chapterHasANumber} and not bool{chapterHasAStar}}%
    {%
        \begin{tikzpicture}[overlay, remember picture]%
            \node [right = \myinnermargin, below = \mytopmargin, inner sep = 0, outer sep = 0, anchor = north west] (numbernode) at (current page.north west)%
            {%
                \hspace{\myinnermargin}%
                \sffamily\fontsize{60}{60}\selectfont%
                \color{chapternumber}%
                \thechapter%
            };%
            \node [inner sep = 0, outer sep = 0, anchor = north west] at (numbernode.south west)%
            {%
                \begin{tikzpicture}[overlay, remember picture]%
                    \draw[color = stroke1, line width = 0.7 mm] (\myinnermargin, -1 ex) -- (\paperwidth, -1 ex);%
                \end{tikzpicture}%
            };%
            \node (align) [text width = \textwidth - 2 cm, align = right, right = \myinnermargin + \mybodywidth, inner sep = 0, outer sep = 0, anchor = east] at (numbernode.west)%
            {%
                #3%
            };%
        \end{tikzpicture}%
    }%
    {%
        \begin{tikzpicture}[overlay, remember picture]%
            \node [right = \myinnermargin, below = \mytopmargin, inner sep = 0, outer sep = 0, anchor = north west] (numbernode) at (current page.north west)%
            {%
                \hspace{\myinnermargin}%
                \sffamily\fontsize{60}{60}\selectfont%
                \color{white}%
                \thechapter%
            };%
            \node [inner sep = 0, outer sep = 0, anchor = north west] at (numbernode.south west)%
            {%
                \begin{tikzpicture}[overlay, remember picture]%
                    \draw[color = stroke1, line width = 0.7 mm] (\myinnermargin, -1 ex) -- (\paperwidth, -1 ex);%
                \end{tikzpicture}%
            };%
            \node (align) [align = left, right = \myinnermargin, inner sep = 0, outer sep = 0, anchor = south west] at (numbernode.south west)%
            {%
                #3%
            };%
        \end{tikzpicture}%
    }%
}
\RedeclareSectionCommand%
[%
    font = \color{stroke1}\normalfont\huge\sffamily,
    afterskip = 20 pt,
]
{chapter}

\BeforeStartingTOC[toc]{\pagestyle{plain}}
\AfterStartingTOC{\thispagestyle{plain}}                        
\usepackage
[
    sortcites,              
    style = alphabetic,     
    defernumbers,           
    safeinputenc,           
    backref = true,         
    backrefstyle = three,   
    hyperref = true,        
    maxbibnames = 99,       
    maxcitenames = 2,       
]
{biblatex} 

\renewbibmacro{in:}%
{%
    \ifentrytype{article}{}{\printtext{\bibstring{in}\intitlepunct}}%
}

\renewbibmacro*{volume+number+eid}%
{%
    \printfield{volume}%
    \iffieldundef{number}{}{\addcolon}%
    \printfield{number}%
    \setunit*{\addcomma\space}%
    \printfield{eid}%
}

\DefineBibliographyStrings{english}%
{%
    backrefpage  = {\lowercase{s}ee page}, 
    backrefpages = {\lowercase{s}ee pages} 
}

\DeclareFieldFormat[article]{title}{\textbf{\color{stroke1}#1}}
\DeclareFieldFormat[inproceedings]{title}{\textbf{\color{stroke1}#1}}
\DeclareFieldFormat[thesis]{title}{\textbf{\color{stroke1}#1}}
\DeclareFieldFormat[book]{title}{\textbf{\color{stroke1}#1}}
\DeclareFieldFormat[unpublished]{title}{\textbf{\color{stroke1}#1}}
\DeclareFieldFormat[report]{title}{\textbf{\color{stroke1}#1}}
\DeclareFieldFormat[inbook]{chapter}{\textbf{\color{stroke1}#1}}
\DeclareFieldFormat[inbook]{title}{#1}
\DeclareFieldFormat{pages}{#1}

\newtoggle{authorend}
\togglefalse{authorend}

\DeclareBibliographyDriver{article}%
{%
  \usebibmacro{bibindex}%
  \usebibmacro{begentry}%
  \iftoggle{authorend}{}{\usebibmacro{author/translator+others}}%
  \setunit{\labelnamepunct}\newblock
  \usebibmacro{title}%
  \newunit
  \printlist{language}%
  \newunit\newblock
  \usebibmacro{byauthor}%
  \newunit\newblock
  \usebibmacro{bytranslator+others}%
  \newunit\newblock
  \printfield{version}%
  \newunit\newblock
  \usebibmacro{in:}%
  \usebibmacro{journal+issuetitle}%
  \newunit
  \usebibmacro{byeditor+others}%
  \newunit
  \usebibmacro{note+pages}%
  \newunit\newblock
  \iftoggle{bbx:isbn}
  {\printfield{issn}}
  {}%
  \newunit\newblock
  \usebibmacro{doi+eprint+url}%
  \newunit\newblock
  \usebibmacro{addendum+pubstate}%
  \setunit{\bibpagerefpunct}\newblock
  \usebibmacro{pageref}%
  \newunit\newblock
  \iftoggle{bbx:related}
  {\usebibmacro{related:init}%
    \usebibmacro{related}}
  {}%
  \usebibmacro{finentry}%
  \iftoggle{authorend}{\usebibmacro{author/translator+others}}{}%
}

\DeclareBibliographyDriver{inbook}%
{%
  \usebibmacro{bibindex}%
  \usebibmacro{begentry}%
  \iftoggle{authorend}{}{\usebibmacro{author/translator+others}}%
  \setunit{\labelnamepunct}\newblock
  \usebibmacro{chapter+pages}%
  \newunit
  \printlist{language}%
  \newunit\newblock
  \usebibmacro{byauthor}%
  \newunit\newblock
  \usebibmacro{in:}%
  \usebibmacro{bybookauthor}%
  \newunit\newblock
  \usebibmacro{maintitle+booktitle}%
  \newunit\newblock
  \usebibmacro{byeditor+others}%
  \newunit\newblock
  \printfield{edition}%
  \newunit
  \iffieldundef{maintitle}
  {\printfield{volume}%
    \printfield{part}}
  {}%
  \newunit
  \printfield{volumes}%
  \newunit\newblock
  \usebibmacro{series+number}%
  \newunit\newblock
  \printfield{note}%
  \newunit\newblock
  \usebibmacro{publisher+location+date}%
  \newunit\newblock
  \newunit\newblock
  \iftoggle{bbx:isbn}
  {\printfield{isbn}}
  {}%
  \newunit\newblock
  \usebibmacro{doi+eprint+url}%
  \newunit\newblock
  \usebibmacro{addendum+pubstate}%
  \setunit{\bibpagerefpunct}\newblock
  \usebibmacro{pageref}%
  \newunit\newblock
  \iftoggle{bbx:related}
  {\usebibmacro{related:init}%
    \usebibmacro{related}}
  {}%
  \usebibmacro{finentry}%
  \iftoggle{authorend}{\usebibmacro{author/translator+others}}{}%
}

\DeclareBibliographyDriver{inproceedings}%
{%
  \usebibmacro{bibindex}%
  \usebibmacro{begentry}%
  \iftoggle{authorend}{}{\usebibmacro{author/translator+others}}%
  \setunit{\labelnamepunct}\newblock
  \usebibmacro{title}%
  \newunit
  \printlist{language}%
  \newunit\newblock
  \usebibmacro{byauthor}%
  \newunit\newblock
  \usebibmacro{in:}%
  \usebibmacro{maintitle+booktitle}%
  \newunit\newblock
  \usebibmacro{event+venue+date}%
  \newunit\newblock
  \usebibmacro{byeditor+others}%
  \newunit\newblock
  \iffieldundef{maintitle}
  {\printfield{volume}%
    \printfield{part}}
  {}%
  \newunit
  \printfield{volumes}%
  \newunit\newblock
  \usebibmacro{series+number}%
  \newunit\newblock
  \printfield{note}%
  \newunit\newblock
  \printlist{organization}%
  \newunit
  \usebibmacro{publisher+location+date}%
  \newunit\newblock
  \usebibmacro{chapter+pages}%
  \newunit\newblock
  \iftoggle{bbx:isbn}
  {\printfield{isbn}}
  {}%
  \newunit\newblock
  \usebibmacro{doi+eprint+url}%
  \newunit\newblock
  \usebibmacro{addendum+pubstate}%
  \setunit{\bibpagerefpunct}\newblock
  \usebibmacro{pageref}%
  \newunit\newblock
  \iftoggle{bbx:related}
  {\usebibmacro{related:init}%
    \usebibmacro{related}}
  {}%
  \usebibmacro{finentry}%
  \iftoggle{authorend}{\usebibmacro{author/translator+others}}{}%
}

\DeclareBibliographyDriver{thesis}%
{%
  \usebibmacro{bibindex}%
  \usebibmacro{begentry}%
  \iftoggle{authorend}{}{\usebibmacro{author}}%
  \setunit{\labelnamepunct}\newblock
  \usebibmacro{title}%
  \newunit
  \printlist{language}%
  \newunit\newblock
  \usebibmacro{byauthor}%
  \newunit\newblock
  \printfield{note}%
  \newunit\newblock
  \printfield{type}%
  \newunit
  \usebibmacro{institution+location+date}%
  \newunit\newblock
  \usebibmacro{chapter+pages}%
  \newunit
  \printfield{pagetotal}%
  \newunit\newblock
  \iftoggle{bbx:isbn}
  {\printfield{isbn}}
  {}%
  \newunit\newblock
  \usebibmacro{doi+eprint+url}%
  \newunit\newblock
  \usebibmacro{addendum+pubstate}%
  \setunit{\bibpagerefpunct}\newblock
  \usebibmacro{pageref}%
  \newunit\newblock
  \iftoggle{bbx:related}
  {\usebibmacro{related:init}%
    \usebibmacro{related}}
  {}%
  \usebibmacro{finentry}%
  \iftoggle{authorend}{\usebibmacro{author}}{}%
}

\providebool{bbx:subentry}
\newbibmacro*{citenum}%
{
  \printtext[bibhyperref]{
    \printfield{prefixnumber}%
    \printfield{labelnumber}%
    \ifbool{bbx:subentry}
    {\printfield{entrysetcount}}
    {}}%
}

\DeclareCiteCommand{\conline}[\mkbibbrackets]
{\usebibmacro{prenote}}
{\usebibmacro{citeindex}%
  \usebibmacro{citenum}}
{\multicitedelim}
{\usebibmacro{postnote}}       
\usepackage
[
    babel = true, 
]
{microtype}           
\usepackage{csquotes} 

\usepackage{amsmath}
\usepackage{amssymb}
\usepackage{amsthm}
\usepackage{thmtools}
\usepackage{mathtools}
\usepackage{thm-restate}
\usepackage{dsfont}        
\usepackage{braceMnSymbol} 

\usepackage
[
    ttscale = 0.85, 
]
{libertine} 
\usepackage
[
    libertine,    
    slantedGreek, 
    vvarbb,       
    libaltvw,     
]
{newtxmath} 
\usepackage{url} 
\usepackage{bm}  

\usepackage{graphicx} 
\usepackage
[
    subrefformat = simple, 
    labelformat = simple,  
]
{subcaption}         
\usepackage{wrapfig} 

\columnsep = \mymargininnersep

\usepackage{array}     
\usepackage{booktabs}  
\usepackage{longtable} 
\usepackage{pdflscape} 

\usepackage{enumitem} 

\usepackage
[
    ruled,         
    vlined,        
    linesnumbered, 
]
{algorithm2e} 

\usepackage{xspace}   
\usepackage
[
    shortcuts, 
]
{extdash}             
\usepackage{setspace} 

\usepackage{xparse}    
\usepackage{footnote}  
\usepackage{afterpage} 
\usepackage
[
    textsize = scriptsize, 
]
{todonotes}            

\usepackage{amssymb}
\usepackage{float}
\usepackage{multirow}

\usepackage
[
    bookmarks = true,                 
    bookmarksopen = false,            
    bookmarksnumbered = true,         
    pdfstartpage = 1,                 
    pdftitle = {{\printTitle}},       
    pdfauthor = {{\printAuthor}},     
    pdfsubject = {{\printSubject}},   
    pdfkeywords = {{\printKeywords}}, 
    breaklinks = true,                
    \ifprintVersion
        hidelinks,                    
    \else
    colorlinks = true,            
    allcolors = stroke1,          
    \fi
]
{hyperref} 

\usepackage
[
    noabbrev,   
    nameinlink, 
]
{cleveref} 
\newcommand*{\colloquialDegreeName}{Master}
\newcommand*{\colloquialDegreeNameLowercase}{master}

\newcommand*{\degreeAbbreviation}{M.}

\ifbachelorThesis
    \renewcommand*{\colloquialDegreeName}{Bachelor}
    \renewcommand*{\colloquialDegreeNameLowercase}{bachelor}
    \renewcommand*{\degreeAbbreviation}{B.}
\fi

\makeatletter
    \def\IfEmptyTF#1%
    {%
        \if\relax\detokenize{#1}\relax%
            \expandafter\@firstoftwo%
        \else%
            \expandafter\@secondoftwo%
        \fi%
    }
\makeatother

\NewDocumentCommand{\mathOrText}{m}
{%
    \ensuremath{#1}\xspace%
}

\let\originalleft\left
\let\originalright\right
\renewcommand{\left}{\mathopen{}\mathclose\bgroup\originalleft}
\renewcommand{\right}{\aftergroup\egroup\originalright}

\makeatletter
    \DeclareRobustCommand{\bfseries}%
    {%
        \not@math@alphabet\bfseries\mathbf%
        \fontseries\bfdefault\selectfont%
        \boldmath%
    }
\makeatother

\xspaceaddexceptions{]\}}

\allowdisplaybreaks

\crefname{ineq}{inequality}{inequalities}
\creflabelformat{ineq}{#2{\upshape(#1)}#3} 

\crefname{term}{term}{terms}
\creflabelformat{term}{#2{\upshape(#1)}#3}

\let\oldfootnote\footnote

\newlength{\spaceBeforeFootnote} 
\newlength{\spaceAfterFootnote}  

\RenewDocumentCommand{\footnote}{o o o m}%
{%
    \IfNoValueTF{#1}%
    {%
        \oldfootnote{#4}%
    }%
    {%
        \setlength{\spaceBeforeFootnote}{\IfEmptyTF{#1}{0}{#1} em}%
        \IfNoValueTF{#2}%
        {%
            \hspace*{\spaceBeforeFootnote}\oldfootnote{#4}%
        }%
        {%
            \setlength{\spaceAfterFootnote}{\IfEmptyTF{#2}{0}{#2} em}%
            \hspace*{\spaceBeforeFootnote}\IfNoValueTF{#3}{\oldfootnote{#4}}{\oldfootnote[#3]{#4}}\hspace*{\spaceAfterFootnote}%
        }%
    }%
}

\makesavenoteenv{figure}
\makesavenoteenv{table}
\makesavenoteenv{tabular}

\iffancyTheorems
    \declaretheoremstyle
    [
        spaceabove = \topsep,
        spacebelow = \topsep,
        headfont = \bfseries,
        headformat = \textcolor{stroke1}{$\blacktriangleright$} \NAME~\NUMBER \NOTE,
        notefont = \bfseries,
        notebraces = {(}{)},
        bodyfont = \normalfont,
        postheadspace = 0.5 em,
        qed = \textcolor{stroke1}{\bfseries$\blacktriangleleft$},
    ]
    {myTheoremStyle}
    
    \declaretheorem
    [
        style = myTheoremStyle,
        name = Conjecture,
        numberwithin = chapter,
    ]
    {conjecture}
    \declaretheorem
    [
        style = myTheoremStyle,
        name = Proposition,
        sharenumber = conjecture,
    ]
    {proposition}
    \declaretheorem
    [
        style = myTheoremStyle,
        name = Claim,
        sharenumber = conjecture,
    ]
    {claim}
    \declaretheorem
    [
        style = myTheoremStyle,
        name = Lemma,
        sharenumber = conjecture,
    ]
    {lemma}
    \declaretheorem
    [
        style = myTheoremStyle,
        name = Corollary,
        sharenumber = conjecture,
    ]
    {corollary}
    \declaretheorem
    [
        style = myTheoremStyle,
        name = Theorem,
        sharenumber = conjecture,
    ]
    {theorem}
    \declaretheorem
    [
        style = myTheoremStyle,
        name = Definition,
        sharenumber = conjecture,
    ]
    {definition}
    \declaretheorem
    [
        style = myTheoremStyle,
        name = Example,
        sharenumber = conjecture,
    ]
    {example}
    \declaretheorem
    [
        style = myTheoremStyle,
        name = Remark,
        sharenumber = conjecture,
    ]
    {remark}
\else
    \theoremstyle{plain}

\fi

\NewDocumentCommand{\functionTemplate}{m m m m o}%
{%
    \IfNoValueTF{#5}%
    {%
        \mathOrText{#1\left#2{#4}\right#3}%
    }%
    {%
        \mathOrText{#1#5#2{#4}#5#3}%
    }%
}

\newcommand*{\leftBracketType}{(}
\newcommand*{\rightBracketType}{)}

\NewDocumentCommand{\createFunction}{m m o o}%
{%
    \renewcommand*{\leftBracketType}{\IfNoValueTF{#3}{(}{#3}}%
    \renewcommand*{\rightBracketType}{\IfNoValueTF{#4}{)}{#4}}%
    \NewDocumentCommand{#1}{o o}%
    {%
        \IfNoValueTF{##1}%
        {%
            \mathOrText{#2}%
        }%
        {%
            \functionTemplate{#2}{\leftBracketType}{\rightBracketType}{##1}[##2]%
        }%
    }%
}

\DeclareDocumentCommand{\probabilisticFunctionTemplate}{m m O{} o}
{%
    \functionTemplate{#1}%
    {\lbrack}%
    {\rbrack}%
    {#2\IfEmptyTF{#3}{}{\ \IfNoValueTF{#4}{\left}{#4}\vert\ \vphantom{#2}#3\IfNoValueTF{#4}{\right.}{}}}%
    [#4]%
}

\ifboldNumberSets

    \newcommand*{\indicatorFunctionSymbol}{\mathbf{1}}
\else

    \newcommand*{\indicatorFunctionSymbol}{\mathds{1}}
\fi

\RenewDocumentCommand{\Pr}{m O{} o}%
{%
    \probabilisticFunctionTemplate{\mathrm{Pr}}{#1}[#2][#3]%
}

\NewDocumentCommand{\E}{m O{} o}%
{%
    \probabilisticFunctionTemplate{\mathrm{E}}{#1}[#2][#3]%
}

\NewDocumentCommand{\Var}{m O{} o}%
{%
    \probabilisticFunctionTemplate{\mathrm{Var}}{#1}[#2][#3]%
}

\DeclareDocumentCommand{\bigO}{m o}%
{%
    \functionTemplate{\mathrm{O}}{(}{)}{#1}[#2]%
}

\DeclareDocumentCommand{\smallO}{m o}%
{%
    \functionTemplate{\mathrm{o}}{(}{)}{#1}[#2]%
}

\DeclareDocumentCommand{\bigTheta}{m o}%
{%
    \functionTemplate{\upTheta}{(}{)}{#1}[#2]%
}

\DeclareDocumentCommand{\bigOmega}{m o}%
{%
    \functionTemplate{\upOmega}{(}{)}{#1}[#2]%
}

\DeclareDocumentCommand{\smallOmega}{m o}%
{%
    \functionTemplate{\upomega}{(}{)}{#1}[#2]%
}

\DeclareDocumentCommand{\eulerE}{o}%
{%
    \mathOrText{\mathrm{e}\IfNoValueTF{#1}{}{^{#1}}}%
}

\DeclareDocumentCommand{\poly}{m o}%
{%
    \functionTemplate{\mathrm{poly}}{(}{)}{#1}[#2]%
}

\createFunction{\id}{\mathrm{id}}

\NewDocumentCommand{\ind}{m o o}%
{%
    \IfNoValueTF{#2}%
    {%
        \mathOrText{\indicatorFunctionSymbol_{#1}}%
    }%
    {%
        \functionTemplate{\indicatorFunctionSymbol_{#1}}{(}{)}{#2}[#3]%
    }%
}

\DeclareDocumentCommand{\dom}{m o}%
{%
    \functionTemplate{\mathrm{dom}}{(}{)}{#1}[#2]%
}

\DeclareDocumentCommand{\rng}{m o}%
{%
    \functionTemplate{\mathrm{rng}}{(}{)}{#1}[#2]%
}

\DeclareDocumentCommand{\d}{o}%
{%
    \mathrm{d}\IfNoValueTF{#1}{}{^{#1}}%
}

\DeclareDocumentCommand{\set}{m m o}%
{
    \mathOrText{\IfNoValueTF{#3}{\left}{#3}\{#1\ \IfNoValueTF{#3}{\left}{#3}\vert\
    \vphantom{#1}#2\IfNoValueTF{#3}{\right.}{}\IfNoValueTF{#3}{\right}{#3}\}}
}      

\begin{document}

    \frontmatter

\ifprintVersion
    \ifprofessionalPrint
        \newgeometry
        {
            textwidth = 134 mm,
            textheight = 220 mm,
            top = 38 mm + \extraborderlength,
            inner = 38 mm + \mybindingcorrection + \extraborderlength,
        }
    \else
        \newgeometry
        {
            textwidth = 134 mm,
            textheight = 220 mm,
            top = 38 mm,
            inner = 38 mm + \mybindingcorrection,
        }
    \fi
\else
    \newgeometry
    {
        textwidth = 134 mm,
        textheight = 220 mm,
        top = 38 mm,
        inner = 38 mm,
    }
\fi

\begin{titlepage}
    \sffamily
    \begin{center}
        \includegraphics[height = 3.2 cm]{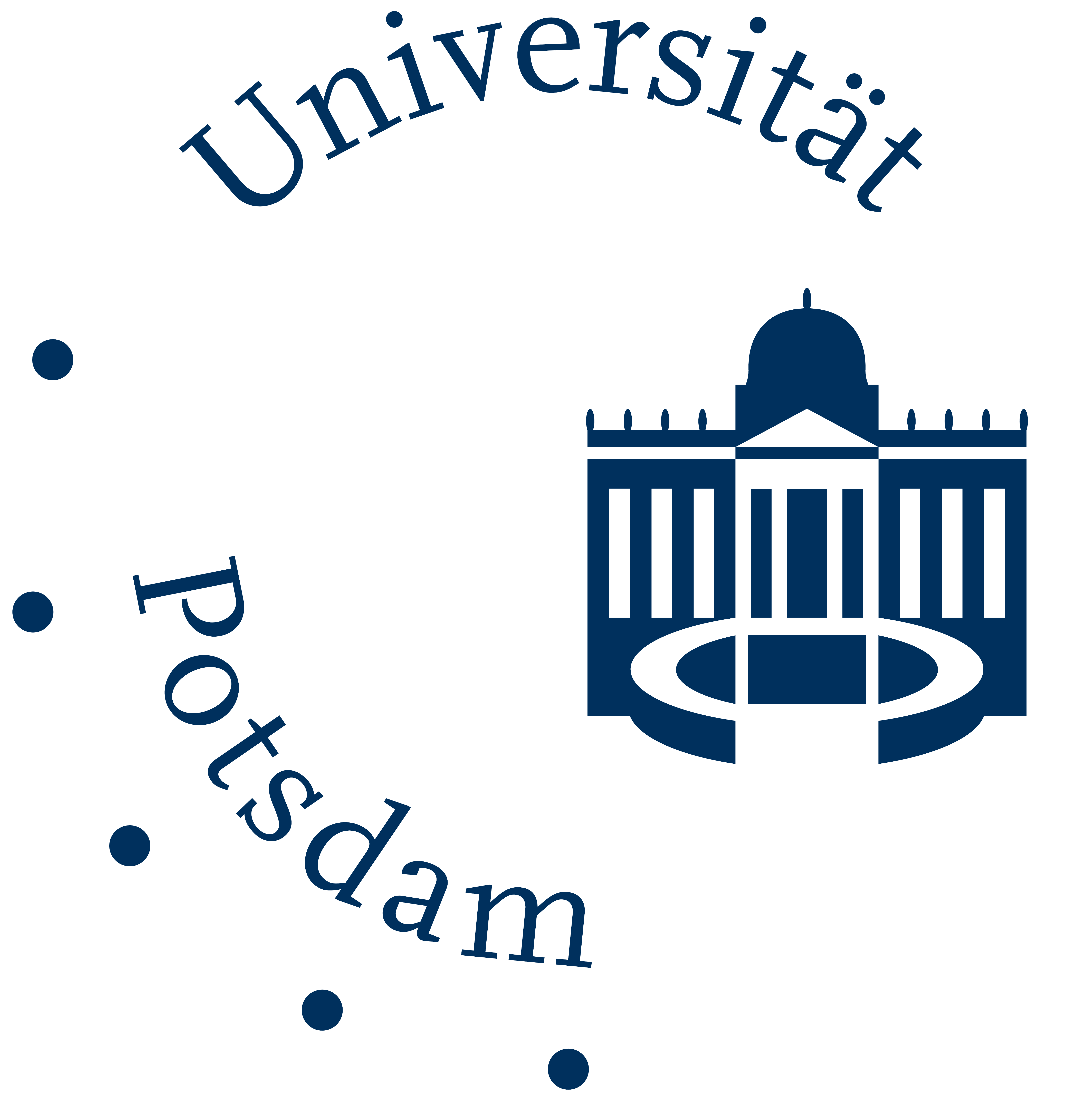} \hfill \includegraphics[height = 3 cm]{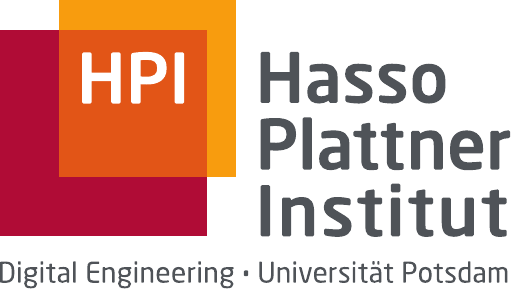}\\
        \vfil
        {\LARGE
            \rule[1 ex]{\textwidth}{1.5 pt}
            \onehalfspacing\printTitleBold\\[1 ex]
            {\vspace*{-1 ex}\Large \printGermanTitle}\\
            \rule[-1 ex]{\textwidth}{1.5 pt}
        }
        \vfil
        {\Large\textbf{\printAuthor}}
        \vfil
        {\large Universitäts\colloquialDegreeNameLowercase arbeit\\[0.25 ex]
        zur Erlangung des akademischen Grades}\\[0.25 ex]
        \bigskip
        {\Large \colloquialDegreeName{} of Science}\\[0.5 ex]
        {\large\emph{(\degreeAbbreviation\,Sc.)}}\\
        \bigskip
        {\large im Studiengang\\[0.25 ex]
        \printProgram}
        \vfil
        {\large eingereicht am \printDateReceived{} am\\[0.25 ex]
        Fachgebiet Algorithm Engineering der\\[0.25 ex]
        Digital-Engineering-Fakultät\\[0.25 ex]
        der Universität Potsdam}
    \end{center}
    
    \vfil
    \begin{table}[h]
        \centering
        \large
        \sffamily 
        {\def\arraystretch{1.2}
            \begin{tabular}{>{\bfseries}p{3.8 cm}p{5.3 cm}}
                Gutachter               & \printNameOfSupervisor\\
                Betreuer                & \printAdditionalExaminers
            \end{tabular}
        }
    \end{table}
\end{titlepage}

\restoregeometry

    \pagestyle{plain}

    \addchap{Abstract}
    The strategic selection of resources by selfish agents has long been a key area of research, with Resource Selection Games and Congestion Games serving as prominent examples. In these traditional frameworks, agents choose from a set of resources, and their utility depends solely on the number of other agents utilizing the same respective resource, treating all agents as indistinguishable or anonymous. Only recently, the study of the Resource Selection Game with heterogeneous agents has begun, meaning agents have a type and the fraction of agents of their type at their resource is the basis of their decision-making. \\
In this work, we initiate the study of the Resource Selection Game with heterogeneous agents in combination with single-peaked utility functions, as some research suggests that this may represent human decision-making in certain cases.\\
We conduct a comprehensive analysis of the game's stability within this framework.
We provide tight bounds that specify for which peak values equilibria exist across different dynamics on cycles and binary trees. 
On arbitrary graphs, in a setting where agents lack information about the selection of other agents, we provide tight bounds for the existence of equilibria, given that the utility function is linear on both sides of the peak. Agents possessing this information on arbitrary graphs creates the sole case where our bounds are not tight, instead, we narrow down the cases in which the game may admit equilibria and present how several conventional approaches fall short in proving stability.

    \selectlanguage{ngerman}
    \addchap{Zusammenfassung}
    Die strategische Auswahl von Ressourcen durch Agenten, die versuchen, ihre individuelle Zufriedenheit zu maximieren, ist seit langem ein zentrales Forschungsthema, beispielsweise in Form von ‘Resource Selection Games’ und ‘Congestion Games’. In diesen traditionellen Modellen wählen Agenten aus einer Menge von Ressourcen, wobei ihre Zufriedenheit an einer Ressource ausschließlich von der Anzahl der Agenten abhängt, die dieselbe Ressource nutzen. Dabei werden alle Agenten als ununterscheidbar oder anonym betrachtet. Erst kürzlich begann die Untersuchung des ‘Resource Selection Games’ mit heterogenen Agenten, was bedeutet, dass Agenten einen Typ haben und der Anteil der Agenten ihres Typs an einer Ressource die Grundlage ihrer Evaluation jener Ressource bildet. \\
In dieser Arbeit initiieren wir die Untersuchung des ‘Resource Selection Games’ mit heterogenen Agenten in Kombination mit single-peaked Präferenzen, da einige Forschungsergebnisse darauf hindeuten, dass dies in bestimmten Szenarien den menschlichen Entscheidungsprozess akkurater modelliert. \\
Wir führen eine umfassende Analyse der Stabilität des Spiels in diesem Rahmen durch. Dabei geben wir bezogen auf $\Lambda$ exakte Schranken für die Stabilität des Spiels auf Zyklen und Binärbäumen in verschiedenen Dynamiken an.  Für beliebige Graphen beweisen wir exakte Schranken für die Existenz von Equilibria, wenn Agenten nicht über die genaue Wahl der anderen Agenten Bescheid wissen und die utility-Funktion auf beiden Seiten des Maximums linear ist. Im Fall, dass Agenten über die genaue Wahl der anderen Agenten Bescheid wissen, demonstrieren wir, wie mehrere herkömmliche Ansätze daran scheitern, die Stabilität zu beweisen, und schränken den Bereich ein, in dem Equilibria existieren könnten.
    \selectlanguage{american}

    \addchap{Acknowledgments}
    I would like to express my sincere gratitude to all those who supported and accompanied me throughout the writing of this thesis. \\
First, working alongside Jannes, Tim, Tobi, and Helena made the experience not only more productive but also more enjoyable. Special thanks to Jannes for helping me discuss some trickier parts. I also extend my appreciation to my friend Jonas and my parents for proofreading parts of this thesis. \\
I am especially grateful to my supervisors, Pascal and Simon, whose invaluable feedback not only shaped this thesis but also helped me keep my findings in perspective. Their constant support was crucial in bringing together my research into a cohesive whole.

    \setuptoc{toc}{totoc}
    \tableofcontents

    \pagestyle{headings}
    \mainmatter

    \chapter{Introduction}
    \newcommand{\G}{\mathbf{G}}
\newcommand{\res}{\mathbf{Q}}
\newcommand{\act}{\mathbf{A}}
\newcommand{\edg}{\mathbf{E}}
\newcommand{\red}{\mathbf{R}}
\newcommand{\blue}{\mathbf{B}}
\newcommand{\strat}{\mathbf{s}}

The study of resource selection in multi-agent environments has a long history across disciplines such as Artificial Intelligence, Operations Research, and Theoretical Computer Science. As such, resources may encompass a broad range of entities, e.g., facilities like schools or hospitals, printers or compute servers. 
The Resource Selection Game serves as a game theoretic adaptation of the above, providing a framework for the analysis of many real-world scenarios. Depending on the specific scenario we aim to model, this framework can be modified accordingly. One notable variation is when the society of agents is heterogeneous. In this variation, similar to the Schelling Game based on Schelling's model for residential segregation \cite{schelling_dynamic_1971}, agents are assigned types. This type then plays a pivotal role in the resource selection process, opposed to other models in which agents only care about the total number of agents at a certain resource. \\
One example for a process that abides to the rules of the Resource Selection Game with a heterogeneous society, would be families choosing a school in their neighborhood for their child. In doing so, they may consider specific attributes of the schools, such as the representation of their ethnic group to a certain extent \cite{burgess_parallel_2005, hailey_racial_2022}. In this case, the schools are the resources from which the families (agents) make their selection. Two key observations in this scenario, which the Resource Selection Game framework with heterogeneous types captures, are that 
\begin{enumerate}
    \item families evaluate a school based on the other families, who have chosen it (specifically, based on their type) and
    \item not every family has access to every school (e.g., due to the distance being too large).
\end{enumerate}
Moreover, we assume that all agents share the same preference for a specific fraction of their type being represented at their chosen resource. By analyzing the Resource Selection Game as an abstraction of this process, we can derive insights into the stability and the potential segregation which occurs in scenarios like the one described above.\\
Additionally, there are different paradigms by which agents may decide to choose a certain resource based on agents' types. One common assumption has been that agents try to maximize the number of agents of their type at their resource. However, Social Science survey studies suggest that this assumption should be challenged \cite{smith_general_2019}. As such, game theoretic models have been studied in combination with several models of the agents' preference. In line with that, we expand the research on the Resource Selection Game with a heterogeneous society by studying it in combination with single-peaked utility functions. Concretely, we base our analysis on the general model where agents have an ideal fraction of same-type agents at their resource that is  in the interval $(0,1)$.

\section{Related Work}
    The problem of selecting resources is a classical topic in combinatorial optimization, with many variations exemplified by problems such as scheduling, packing, and covering \cite{papadimitriou_combinatorial_2013}. The study of strategic resource selection was pioneered with the study of Congestion Games \cite{rosenthal_class_1973}, where agents choose from a predefined set of resources, and their associated cost depends exclusively on the number of agents sharing the selected resources. Subsequent developments have introduced weighted variants and agent-specific cost functions, allowing for more complex scenarios \cite{milchtaich_congestion_1996}. Notable applications include strategic path selection in networks \cite{roughgarden_how_2002, anshelevich_price_2004} and the problem of selfish job scheduling on servers \cite{tardos_selfish_2007}. Additionally, competitive facility location models, where either facilities vie for clients \cite{vetta_nash_2002} or clients compete for access to facilities \cite{kohlberg_equilibrium_1983, peters_hotellings_2018, krogmann_two-stage_2021, krogmann_strategic_2023}, can be viewed as another form of strategic resource selection. Group activity selection provides another example, where agents choose activities based on their preferences and the participation of others \cite{darmann_group_2012, igarashi_group_2017}. Across all these frameworks, an agent’s utility or cost is typically influenced by how many other agents are utilizing the same resources. \\
    In contrast to this, our work bears more resemblance to models that involve heterogeneous agents. Recently, game-theoretic approaches have explored the formation of networks by homophilic agents \cite{bullinger_network_2022}. One of the most pertinent models to our study is Schelling’s model of residential segregation \cite{schelling_dynamic_1971}, where agents of different types strategically choose locations in a residential area. These agents follow a threshold-based utility function, achieving maximum utility when at least a certain fraction of neighbors shares their type. Game-theoretic extensions of Schelling’s model, known as Schelling Games, have been the subject of recent investigations \cite{chauhan_schelling_2018, echzell_convergence_2019}. Moreover, variations of these games, where agents aim to maximize the proportion of same-type neighbors, have gained prominence \cite{agarwal_schelling_2021, bullinger_welfare_2021, kanellopoulos_modified_2021, kanellopoulos_not_2022, bilo_topological_2022}.
    \\
    However, a key distinction remains: in Schelling Games, each resource (i.e., location) can be chosen by at most one agent, meaning that agents’ neighborhoods only partially overlap. Furthermore, the size of these neighborhoods is constrained by the structure of the underlying graph representing the residential area. In contrast, our model allows resources to be shared by an arbitrary number of agents without such fixed neighborhood boundaries, which introduces different dynamics in the strategic selection process.
    \\
    Closely related to our work are Hedonic Diversity Games \cite{bredereck_hedonic_2019, boehmer_individual-based_2020, darmann_hedonic_2021, ganian_hedonic_2022}, where agents of varying types strategically form coalitions and their utility is determined by the proportion of same-type agents within their chosen coalition. While Hedonic Diversity Games allow for individual preferences of agents, our model extends the case in which all agent preferences are identical. In our framework, access to resources can be restricted, effectively generalizing these special cases of Hedonic Diversity Games. Additionally, Hedonic Expertise Games \cite{caskurlu_hedonic_2021}, where agents’ utility increases with the diversity of types in their coalition, have been explored.
    \\
    In our work, we focus on single-peaked utility functions, a concept rooted in single-peaked preferences, originally introduced by \textcite{black_rationale_1948}. These preferences are well-established in the Economics and Game Theory literature. Notably, single-peaked preferences lead to desirable outcomes in contexts such as Hedonic Diversity Games, as well as in voting and social choice theory \cite{walsh_uncertainty_nodate,yu_multiwinner_nodate,betzler_computation_2013, elkind_characterization_2020, brandt_bypassing_2015}. The above-mentioned Schelling Game has also been investigated in the context of single-peaked utility functions \cite{bilo_tolerance_2022, friedrich_single-peaked_2023}.
    \\
    The Resource Selection Game was formally introduced recently by \textcite{harder_strategic_2023}. To date, research on the Resource Selection Game has exclusively focused on models incorporating a monotonous $\tau$-threshold utility function, similar to those found in Schelling’s model for residential segregation \cite{schelling_dynamic_1971}. We note that for $\tau=1$ the threshold function is equivalent to a single-peaked utility function with a peak at $1$ in our model.

\section{Model}
We consider a strategic game, called the \emph{Resource Selection Game} played on a bipartite \emph{accessibility graph} $\G = (\res \cup \act, \edg)$, where $\res \cap \act = \emptyset$, with $\res$ being the set of \emph{resources} and $\act$ the set of \emph{agents}. A resource $q \in \res$ can be used by an agent $a \in \act$ if and only if $\{q,a\} \in \edg$ and may be used by any number of agents at the same time. We use the short hands $|\act|=n$ and $|\res|=k$.\\

Every agent in the Resource Selection Game is either \emph{red} or \emph{blue}. For the set of red agents $\red$ and the set of blue agents $\blue$ it holds that $\red \cup \blue = \act$ and $\red \cap \blue = \emptyset$. We denote $|\red|=r$ and $|\blue|=b$.\\

Formally, for an agent $a$ we denote their \emph{accessible resources} with $Q(a)\subseteq \res$. Similarly, $A(q)$ defines the agents, who have access to the resource $q$. We call $deg_G(a)=|Q(a)|$ the \emph{degree} of $a$ and $deg_G(q)=|A(q)|$ the \emph{degree} of $q$. As a shorthand for the \emph{maximum degree of any resource in $G$} we use $\Delta_G(\res)=max_{q\in\res}(deg_G(q))$. The selected resource of any agent $a$ in a given step is denoted by $\strat(a)$ and is called $a$'s \emph{strategy}. The vector of all agents' strategies $\strat=(\strat(a_1),\strat(a_2),...,\strat(a_n))$, with $a_i\in\act$ and $s(a_i)\in Q(\strat(a_i))$ for $i\in[1,n]$ is called \emph{strategy profile}. The number of agents at a given resource in a given strategy profile is denoted by $\#(q,\strat)=|\{a\in \act|\strat(a) = q\}|$. Similarly, we define $\#_R(q,\strat)=|\{a\in \red|\strat(a) = q\}|$ and $\#_B(q,\strat)=|\{a\in\blue|\strat(a) = q\}|$. To denote the set of agents at a specific resource $q\in\res$ in a strategy profile $\strat$, we write $A(q,\strat)=\{a\in\act\mid\strat(a)=q\}$. We call the 2-tuple $(\#_R(q,\strat),\#_B(q,\strat))$ the \emph{state} of $q$. We present resource \emph{state transitions} with an arrow from the initial state towards the new state in the color of the agent, who joined or left the resource to cause the state transition (e.g., $\color{red} 1 \color{blue} 1 \color{red} \rightarrow 2 \color{blue} 1$ \color{black} visualizes the state transition of a resource with one blue and one red agent, when a red agent joins). As a shorthand for a resource changing its state and then changing again to return to its initial state we use a double-headed arrow (e.g., $\color{red} 1 \color{blue} 1 \color{red} \leftrightarrow 2 \color{blue} 1$ \color{black} to visualize a resource oscillating between a state with one or two red agents and one blue agent because of alternating joins and leaves by red agents). \\

Additionally, let $\rho_R(q,\strat) = \frac{\#_R(q,\strat)}{\#(q,\strat)}$ and
$\rho_B(q,\strat) = \frac{\#_B(q,\strat)}{\#(q,\strat)}$ denote the fraction of red or blue agents at $q$ in $\strat$. We note that these functions are undefined for empty resources (where $\#(q,\strat)=0$). If either $\#_R(q,\strat)=0$ or $\#_B(q,\strat)=0$ holds, we call $q$ \emph{monochromatic}. In particular, if $q$ is monochromatic in $\strat$ and $\#(q,\strat)=m$, we call q \emph{m-monochromatic}.
\\
 An agent's \emph{utility} is its satisfaction at a given resource in a given strategy profile. We diverge from the Schelling Resource Selection Game in that agents are not necessarily homophilic. Instead, agents prefer a fraction $\Lambda \in (0,1)$ of agents of their color at their resource. The utility of an agent then decreases with the fraction of agents of the same color moving further away from $\Lambda$ in either direction. Because of these properties, we call the function to determine the agents' utility in our model \emph{single-peaked} (around $\Lambda$). To formally define this function, we first define the following function.
\begin{restatable}[The function $p$]{definition}{definition_p}
    The function $p$ has domain $[0,1]$ and exactly one peak at some fraction $\Lambda \in (0,1)$. Additionally, every valid instantiation of $p$ has to fulfill the following properties
    \begin{itemize}
        \item $p$ is monotonically increasing in $[0,\Lambda]$ with $p(0)=0$
        \item for all $x \in (\Lambda,1]$ $p(x) = p(\frac{\Lambda(1-x)}{1-\Lambda})$.
    \end{itemize}
    Note that for all $x \in (\Lambda,1]$, it holds that $\frac{\Lambda(1-x)}{1-\Lambda}\in[0,\Lambda]$ \footnote[-0,2][0]{The statement follows directly from transforming the assumption:
       \begin{equation*} 
           x>\Lambda \iff \frac{1}{1-x}>\frac{1}{1-\Lambda} \iff \frac{\Lambda}{1-x}>\frac{\Lambda}{1-\Lambda} \iff \Lambda>\frac{\Lambda(1-x)}{1-\Lambda}.
       \end{equation*}} 
   and thus $p(x)\in[0,1]$. \
\end{restatable}
\begin{figure}
    \centering
    \makebox[0.5\textwidth][c]{\includegraphics{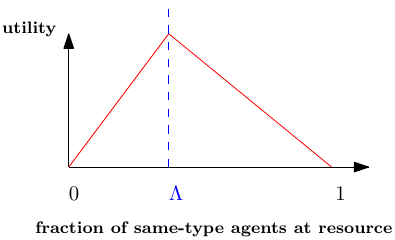}}
    \caption{Example of a linear and single-peaked function $p$.}
    \label{fig:p_demo}
\end{figure}
A visualization of a $p$ function can be seen in \autoref{fig:p_demo}. One specific type of $p$-function, which is a necessary property for our impact-blind main result, is a \emph{linear $p$-function}. A $p$-function is called linear if and only if there exists $m_p\in\mathbb{R}_+$ such that for all $x\in(0,\Lambda]$ it holds that $p(x)=m_p\cdot x$. \\
Using the function $p$ we next define the \emph{utility function} $u$. 
\begin{restatable}[Utility Function]{definition}{definition_utility_function}
    Let $a\in\act$ and $\strat$ be a strategy profile. Then
        \[ 
    u(a,\strat)= \left\{
    \begin{array}{ll} 
        \label{def:p}
          p(\rho_R(\strat(a))) & \text{if } a \in \red \\
          p(\rho_B(\strat(a))) & \text{if } a \in \blue\\
    \end{array} 
    \right. 
    \]
    is $a$'s utility in $\strat$.
\end{restatable}
For a strategy profile $\strat$ with $\strat(a_i) = q$, we use the shorthand $\strat = (q,\bold{s_{-i}})$, where $\bold{s_{-i}} = (\strat(a_1),…\strat(a_{i-1}),\strat(a_{i+1}),…,\strat(a_n))$ is identical to $\strat$ without the $i$-th entry. The \emph{social welfare} of a strategy profile $\strat$ is the sum of all agents' utilities in $\strat$, formally defined as $W(\strat)=\sum_{a\in\act}u(a,\strat)$. Moreover, to denote the social welfare at a specific resource $q\in\res$ within the strategy profile $\strat$, we write $W_q(\strat)=\sum_{a\in A(q,\strat)}u(a,\strat)$.

Additionally, we differentiate between \emph{impact-blind} and \emph{impact-aware} agents. Let $\bold{s'}=(q',\strat_{-i})$ be the result of the jump of agent $i$ from resource $q$ to $q'$, with everything else staying the same, i.e., $\strat(a_j)=\bold{s'}(a_j)$ if and only if $j\neq i$. We call the strategy change from $\strat$ to $\bold{s'}$ an impact-aware improving move for an agent $a_i$ if $u(a_i,(q',\bold{\strat'})) > u(a_i,(q,\strat))$. Meanwhile, an impact-blind improving move for a w.l.o.g. red agent $a_i$ is a strategy change, where $p(\rho_R(q',\strat))>u(a_i,\strat)$. The impact-blindness models a lack of knowledge of agents about the exact number of agents on the other resources, meaning that while agents know the fraction of their type at other resources, they are unaware of how their jump to these resources will alter that resources' fraction. We note that in our model in the impact-blind setting agents never jump to empty resources as for $q\in\res$ with $A(q)=\emptyset$ and a strategy profile $\strat$ the utility they see $p(\rho_R(q,\strat))$ would be undefined. This is in accordance with \autoref{def:p} as $p(0)=0$. We see that a similar statement holds for the impact-aware setting.
\begin{restatable}[]{lemma}{IA_no_jumps_to_empty}
    In the impact-aware setting, agents do not jump to empty resources.
\end{restatable}

\begin{proof}
    Let $q'\in\res$ with $A(q')=\emptyset$. For a jump by an agent $a\in\act$ from a resource $q\in\res$ to $q'$, changing the strategy profile from $\strat$ to $\strat'$, it holds that $u(a,(q',\strat'))=p(1)=0\leq u(a,\strat)$ and thus the move is not impact-aware improving.
\end{proof}

We note that an impact-aware improving move may not be an impact-blind improving move. \footnote[-0,2][0]{For $\Lambda=\frac{1}{2}$ let a red agent $a$ jump from a resource $q$ to $q'$, altering the strategy profile from $\strat$ to $\strat'$, with $u(a,\strat)=p(\frac{2}{5})$, $\#_R(q',\strat)=1$ and $\#_B(q',\strat)=2$. Then it holds that $p(\rho_R(q',\strat))<u(a,\strat)<u(a,\strat')$.}

Furthermore, deviating from the Schelling Resource Selection Game \cite{harder_strategic_2023}, an impact-blind improving move may not be impact-aware improving. \footnote[-0,2][0]{To see this, let $\Lambda=\frac{1}{2}$ and let a red agent $a$ jump from a resource $q$ to $q'$, altering the strategy profile from $\strat$ to $\strat'$, with $u(a,\strat)=p(\frac{2}{5})$, $\#_R(q',\strat)=1$ and $\#_B(q',\strat)=1$. Then it holds that $p(\rho_R(q',\strat))>u(a,\strat)>u(a,\strat')$.}

Based on impact-aware and impact-blind improving moves, we define two states of strategy profiles. We say a strategy profile $\strat$ is in \emph{impact-aware-equilibrium} (\emph{IAE}), if no agent has an impact-aware improving move they wish to make. Similarly, a strategy profile $\strat$ is in \emph{impact-blind-equilibrium} (\emph{IBE}), if no agent has an impact-blind improving move they wish to make.

We say that a game in which agents only make impact-blind improving moves has \emph{impact-blind dynamics} and a game in which agents only perform impact-aware improving moves has \emph{impact-aware dynamics}. In a game with impact-blind dynamics, the \emph{IB-FIP} (\emph{impact-blind finite improvement property}) holds if and only if from every strategy profile an IBE is reached after a finite number of steps. We define the \emph{IA-FIP} analogously.

Note, that we only consider unilateral moves, meaning agents can change their strategy without the consent of other agents being necessary.
\newpage
\section{Our Contribution}

\begin{table}[htbp]
    \centering
    \resizebox{\textwidth}{!}{
    \begin{tabular}{|c|c|c|c|c}
        \hline
        \textbf{Stability Results} &Cycle &Binary Tree &Arbitrary Graphs \\
        \hline
        IB & \checkmark & $\Lambda\in[\frac{3}{5},1)$ (Lemma 3.2) & $\Lambda \geq L_\Delta(\res)$* (Thm. 3.5) \\
        \hline
        IA & \checkmark & $\Lambda\in(0,\frac{1}{2}]$ (Thm. 3.13) & $\Lambda\leq U_\Delta(\res)\textbf{?}$\\
        \hline
    \end{tabular}
    }
    \caption{Comparison of Stability Results: The table presents the intervals of $\Lambda$ for which the FIP holds on the respective graph class in the respective setting. For values of $\Lambda$ outside these intervals, we show that instances exists for which no equilibrium exists. A "\checkmark" indicates that the FIP holds for all $\Lambda \in (0,1)$. A "?" denotes that the FIP \emph{may} hold within this interval, although we still show that for all $\Lambda$ outside the specified range, there exist instances without an equilibrium. A "*" denotes that while the FIP over the specified interval is proven to hold for linear $p$-functions, it remains uncertain whether it holds for non-linear $p$-functions.}
    \label{tab:our_contribution}
\end{table}

In this work, we introduce the study of the Resource Selection Game with single-peaked utility functions, extending recent research on the game with $\tau$-threshold utility functions.

Our findings, with the sole exception being our result for the impact-blind setting on arbitrary graphs, are valid for all functions that satisfy the two simple constraints required by our model in \autoref{def:p}, making the results widely applicable. For the above-mentioned exception, we additionally require $p$ being linear.

We focus on examining the stability of the Resource Selection Game under single-peaked utility functions across three graph classes: cycles, binary trees, and arbitrary graphs. Our results are summarized in \autoref{tab:our_contribution}. We prove that for the $\Lambda$-values listed in the table (except for the one followed by a "?"), the FIP holds in the corresponding setting. It is important to note that from the FIP, it directly follows that an equilibrium exists, which can be computed by following a best response sequence from any strategy profile. However, it is possible for an equilibrium to exist even when the FIP does not hold. Therefore, for all $\Lambda$-values where the FIP does not hold, we additionally prove that there exist instances within the respective graph class, setting and $\Lambda$-range for which no strategy profile in equilibrium exists, therefore making our bounds tight. Specifically, for arbitrary graphs, we distinguish between different $\Lambda$-values based on $\Delta_G(\res)$. This differentiation primarily helps to identify $\Lambda$-values that result in the utility function being either monotonically increasing or decreasing over the range of obtainable agent fractions (except for fractions $0$ and $1$). Concretely, we obtain the lower bound
\begin{equation*}
    L_\Delta(\res)=\frac{\Delta_G(\res)(\Delta_G(\res)-2)}{\Delta_G(\res)^2-\Delta_G(\res)-1}
\end{equation*}
and the upper bound
\begin{equation*}
    U_\Delta(\res)=\frac{\Delta_G(\res)-1}{\Delta_G(\res)^2-\Delta_G(\res)-1}.
\end{equation*}

In the impact-blind setting our main results show that, on binary trees and if $p$ is linear on arbitrary graphs as well, the game possessing the IB-FIP is equivalent to an agent gaining utility if and only if the fraction of agents sharing their color at their resource increases (except for the fraction $1$). In the impact-aware setting, we prove that on binary trees, the game has the IA-FIP if and only if $\Lambda\in(0,\frac{1}{2}]$. On arbitrary graphs, we show that for $\Lambda>U_\Delta(\res)$ instances exist, in which no strategy profile in equilibrium exists. For $\Lambda\leq U_\Delta(\res)$, while we do not provide a conclusion regarding the stability of the game, we share insights into this scenario and show how several conventional approaches to prove stability fail in this case. Besides, we show that on cycle graphs in both settings, impact-blind or impact-aware, the FIP holds for all $\Lambda$.


    \chapter{Preliminaries}
    We start by investigating how $\Lambda$ dictates the relationship of function values of $p$ for different fractions. First, for two fractions, we obtain a lower bound on $\Lambda$ such that a greater fraction is mapped to a greater value by $p$.

\begin{restatable}[]{lemma}{prelim_lemma_1}
\label{lma:lambda_lower_bound_monotonical}
For all $x,y \in (0,1)$ with $x<y$, $p(x)< p(y) \iff \Lambda > \frac{x}{1-y+x}$
\end{restatable}

\begin{proof}
    If $\Lambda \leq x$, it would follow that $p(y)< p(x)$ from $p$ being monotonically decreasing in $[\Lambda,1)$. Similarly, for $\Lambda \geq y$, we have that $p(y)>p(x)$ since $p$ is monotonically increasing in $[0,\Lambda]$. Thus, the greatest lower bound for $\Lambda$ such that $p(x)<p(y)$ must lie in the interval $(x,y)$. It follows from \autoref{def:p} that $p(x) < p(y)$ is equivalent to
    \begin{equation*}
        p\left(\frac{\Lambda(1-y)}{1-\Lambda}\right) > p(x).
    \end{equation*}
    As $p$ is monotonically increasing in $[0,\Lambda]$ the above is equivalent to
    \begin{equation*}
        \frac{\Lambda(1-y)}{1-\Lambda} > x.
    \end{equation*}
    By simplification, we obtain the equivalent equation
    \begin{equation*}
        \Lambda-y\Lambda > x - x\Lambda \iff
        \Lambda(1-y+x) > x \iff
        \Lambda > \frac{x}{1-y+x}.
        \qedhere
    \end{equation*}
\end{proof}

Next, we make the observation that the largest fraction in $(0,1)$ with a denominator of at most $n$ is $\frac{n-1}{n}$.
\begin{restatable}[]{lemma}{lemmaprelimn4}
    \label{lma:largest_fraction}
    For $n \in \mathbb{N}_{\geq2}$ it holds for all $x,y \in \mathbb{N}$, with $x<y\leq n$, that $\frac{n-1}{n}\geq\frac{x}{y}$.
\end{restatable}

\begin{proof}
    From the assumption, it follows that
    \begin{equation*}
        y \geq x+1 \iff
        1 \geq \frac{x+1}{y} \Rightarrow
        1 \geq \frac{1}{n} + \frac{x}{y} \iff
        \frac{x}{y} \leq 1-\frac{1}{n} \iff
        \frac{x}{y} \leq \frac{n-1}{n}.
        \qedhere
    \end{equation*}
\end{proof}
It follows, that in the context of the Resource Selection Game, the largest fraction of a color we may see at any resource $q \in \res$ is $\frac{deg_G(q)-1}{deg_G(q)}$. \\
\autoref{lma:lambda_lower_bound_monotonical} and \autoref{lma:largest_fraction} now enable us to analyze how large $\Lambda$ needs to be for any given graph, for utilities to increase if and only if the fraction of agents' of the same color increases and similarly, providing an upper bound for $\Lambda$, such that agents' utilities improve if and only if the fraction of the agent's color at their resource decreases. We note that moves, which increase an agent's fraction to $1$ or decrease it to $0$ are exceptions from this and are never improving, as both fractions result in a utility of $0$.

\begin{restatable}[]{lemma}{}
    For all graphs $G$ and $\Lambda \geq \frac{\Delta_G(\res)(\Delta_G(\res)-2)}{\Delta_G(\res)^2-\Delta_G(\res)-1}$, for all $x,y\in(0,1)$ with $x<y$ it holds that $p(x)<p(y)$.
    \label{lma:lower_bound_mon_incr_proof}
\end{restatable}
\begin{proof}

    We first observe that since $p$ is monotonically increasing in $(0,\Lambda]$ the equivalence holds for this interval.\\
    To show the equivalence for $(\Lambda,1)$, we first see that at any resource there can be at most $\Delta_G(\res)$ agents at once. Thus, from \autoref{lma:largest_fraction} it follows, that $z_1=\frac{\Delta_G(\res)-1}{\Delta_G(\res)}$ is the largest possible fraction of agents at a resource less than $1$, that may be present. Similarly, to obtain the second-largest possible fraction, we see that \autoref{lma:largest_fraction} gives $z_2=\frac{\Delta_G(\res)-2}{\Delta_G(\res)-1}$ for $\Delta_G(\res)-1$. This is indeed the second-largest fraction, less than $1$ as $\frac{\Delta_G(\res)-2}{\Delta_G(\res)-1}>\frac{x}{\Delta_G(\res)}$ for all $x\in \mathbb{N}$ with $x\leq \Delta_G(\res)-2$.\\
    From \autoref{lma:lambda_lower_bound_monotonical} and our assumption for $\Lambda$, we obtain that $p(z_1)>p(z_2)$. Thus, it must follow that $\Lambda>z_2$, since $p$ is monotonically decreasing in $[\Lambda,1)$. Consequentially, the only fraction possibly larger than $\Lambda$ is $z_1$. If $\Lambda\geq z_1$ the equivalence follows directly from $p$ being monotonically increasing in $(0,\Lambda]$. If $\Lambda<z_1$, it follows from $p(z_1)>p(z_2)$, the fact that $p(z_2)$ is the largest utility for a fraction of agents possible in the game in $(0,\Lambda]$ and the transitivity of $>$ that the equivalence holds.
\end{proof}

We continue with a statement analogous to \autoref{lma:lower_bound_mon_incr_proof}, which provides an upper bound for $\Lambda$ such that larger fractions lead to lower utilities.
\begin{restatable}[]{lemma}{}
    For all graphs $G$ and $\Lambda\leq \frac{\Delta_G(\res)-1}{\Delta_G(\res)^2-\Delta_G(\res)-1}$, for all $x,y\in(0,1)$ with $x<y$ it holds that $p(x)>p(y)$.
    \label{lma:upper_bound_mon_decr_proof}
\end{restatable}
\begin{proof}

    We first observe that since $p$ is monotonically decreasing in $[\Lambda,1]$ the equivalence holds for this interval.\\
    To show the equivalence for $(0,\Lambda)$, we first see that at any resource there can be at most $\Delta_G(\res)$ agents at once. Thus, from \autoref{lma:largest_fraction} it follows, that $z_1=\frac{1}{\Delta_G(\res)}$ is the smallest possible fraction of agents at a resource larger than $0$, that may be present. Similarly, to obtain the second-smallest possible fraction, we see that \autoref{lma:largest_fraction} suggests $z_2=\frac{1}{\Delta_G(\res)-1}$ for $\Delta_G(\res)-1$. This is indeed the second-smallest fraction larger than $0$ as $\frac{1}{\Delta_G(\res)-1}<\frac{x}{\Delta_G(\res)}$ for all $x\in \mathbb{N}$ with $x\geq 2$.\\
    From \autoref{lma:lambda_lower_bound_monotonical} and our assumption for $\Lambda$, we obtain that $p(z_1)>p(z_2)$. Thus, it must follow that $\Lambda<z_2$, since $p$ is monotonically increasing in $(0,\Lambda]$. Consequentially, the only fraction possibly less than $\Lambda$ is $z_1$. If $\Lambda\leq z_1$ the equivalence follows directly from $p$ being monotonically decreasing in $[\Lambda,1)$. If $\Lambda>z_1$, it follows from $p(z_1)>p(z_2)$, the fact that $p(z_2)$ is the smallest utility for a fraction of agents possible in the game in $(\Lambda,1)$ and the transitivity of $>$ that the equivalence holds.
\end{proof}

Moving on, we observe a property, which emerges from the combination of single-peaked utility functions, concretely the fact that $p(0)=0$ and the definition of impact-blindness.

\begin{restatable}[]{lemma}{prelimlemma3}
    \label{lma:IB_monochromatic_no_joins}
    In the impact-blind setting, for any single-peaked utility function, agents do not join monochromatic resources.
\end{restatable}

\begin{proof}
    Let $q$ be a monochromatic resource in a strategy profile $\bold{s}$. Impact-blind agents of the same color see a utility of $p(1)=0$, similarly agents of different colors see a utility of $p(0)=0$, thus no agent will join $q$.
\end{proof}
This property can now be used to prove the following Lemma.
\begin{restatable}[]{lemma}{prelimlemma3}
    \label{lma:IB_monochromatic}
    In the impact-blind setting, for any single-peaked utility function monochromatic resources will change their state only finitely often for the remainder of the game.
\end{restatable}

\begin{proof}
    Let $q$ be a monochromatic resource in a strategy profile $\bold{s}$. According to \autoref{lma:IB_monochromatic_no_joins}, no agent joins $q$. Additionally, according to our model, no agent joins an empty resource. It follows that the state of $q$ changes at most $\#(q,\bold{s})$ times because of agents leaving, until it becomes empty.
\end{proof}

This allows for the construction of potential function arguments using the number of monochromatic resources in the impact-blind setting later on. We further make the observation that from the above \autoref{lma:IB_monochromatic}, resources with a degree of at most $2$ are effectively negligible. 

\begin{restatable}[]{lemma}{lemmaprelimn5}
    \label{lma:IB_deg_leq_2}
    In the impact-blind setting for any resource $q \in \res$ with $deg_G(q) \leq 2$ for all $\Lambda \in (0,1)$ no agent ever jumps to $q$.
\end{restatable}

\begin{proof}
    If $deg_G(q)=0$, no agent can join $q$. If $deg_G(q)=1$, an adjacent agent will not join $q$ as it is empty. If $deg_G(q)=2$ an impact-blind agent will see either an empty resource or $p(1)=0$ and thus never join $q$.
\end{proof}
    \chapter{Main Results}
    We divide our analysis regarding the existence of equilibria in the Resource Selection Game into two categories: impact-blind and impact-aware dynamics. In the impact-blind section, which serves as our starting point, we progressively establish tight bounds on $\Lambda$ regarding the existence of equilibria on arbitrary graphs. In the impact-aware section, our main result focuses on providing tight bounds on $\Lambda$ regarding the existence of equilibria on binary trees.

\section{Impact-Blind Equilibria}
This section provides tight bounds on $\Lambda$ for the existence of equilibria in the Resource Selection Game on arbitrary graphs if $p$ is linear. We progress through several graph classes with increasing complexity, from cycle graphs over binary trees to arbitrary graphs. The analysis of the two former classes allows us to make some observations, which are then encapsulated in the potential argument that we use to obtain our main result in this section about the existence of equilibria on arbitrary graphs. \\
As stated in \autoref{lma:IB_deg_leq_2}, in the impact-blind setting on cycle graphs, no agent makes a move since each agent perceives a utility of 0, given that the maximum degree of any resource is 2. Consequently, the IB-FIP holds for all $\Lambda\in(0,1)$. It is important to highlight that the game’s stability arises from the fact that agents do not jump to empty or monochromatic resources—conditions that, on a cycle graph, apply to all resources. \\
Other than on cycle graphs, on binary trees the game does not admit a stable state for all $\Lambda$ as we will see below.
\begin{restatable}[]{lemma}{lemma_bin_tree_no_IBE}
    There exists a binary tree, such that for all $\Lambda \in \left(0,\frac{3}{5}\right)$ no IBE exists.
    \label{lma:bin_tree_no_IBE}
\end{restatable}
\begin{proof}
    By definition, for any binary tree $G=(\act\cup\res,\edg)$ it holds that $\Delta_G(\res)\leq3$. Since $\Lambda < \frac{3}{5} = \frac{\frac{1}{2}}{1-\frac{2}{3}+\frac{1}{2}}$, from \autoref{lma:lower_bound_mon_incr_proof} it follows that $p\left(\frac{1}{2}\right)>p\left(\frac{2}{3}\right)$. We now construct the graph in \autoref{fig:IBE_counterexample_bin_tree}. \\
    \begin{figure}
        \centering
        \includegraphics[width=0.3\textwidth]{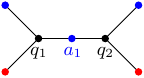}
        \caption{Red dots are red agents, blue dots are blue agents and black dots are resources. An agent $a$ has access to a resource $q$ if and only if $a$ and $q$ are adjacent in the above graph.}
        \label{fig:IBE_counterexample_bin_tree}
    \end{figure}
    We quickly see, that there are only two possible strategy profiles, one where $a_1$ is at $q_1$ and one where $a_1$ is at $q_2$. In both cases, $a_1$ has utility $p\left(\frac{2}{3}\right)$ and sees utility $p\left(\frac{1}{2}\right)$ at the other resource. Thus, from $p\left(\frac{1}{2}\right)>p\left(\frac{2}{3}\right)$ it follows, that $a_1$ alternates between $q_1$ and $q_2$ indefinitely and none of the two possible strategy profiles is in IBE.
\end{proof}
However, for $\Lambda$ values large enough, the IB-FIP still holds.
\begin{restatable}[]{lemma}{lemma_bin_tree_IBE}
    On any binary tree for all $\Lambda \in \left[\frac{6}{10},1\right)$ the IB-FIP holds.
    \label{lma:bin_tree_IB_FIP}
\end{restatable}
\begin{proof}
    Since $\Lambda \in \left[\frac{3}{5},1\right)$ it follows from \autoref{lma:lower_bound_mon_incr_proof} that $p\left(\frac{2}{3}\right) \geq p\left(\frac{1}{2}\right) \geq p\left(\frac{1}{3}\right)$.
     We note that on binary trees in the impact-blind setting, agents only ever see the utilities $0$ and $p\left(\frac{1}{2}\right)$. Thus, the only jumps, which can occur are ones where an agent's utility changes from $0$ to $p\left(\frac{2}{3}\right)$, as the agent sees $p\left(\frac{1}{2}\right)$ and similarly jumps from $p\left(\frac{1}{3}\right)$ to $p\left(\frac{2}{3}\right)$. \\
     A jump, which alters an agent's utility from $p\left(\frac{1}{3}\right)$ to $p\left(\frac{2}{3}\right)$ makes the resource they previously occupied monochromatic. The resource they join had one blue and one red agent before, meaning it was not monochromatic, from which it follows that the total number of monochromatic resources increases by one. From this, it follows, that this type of jump occurs at most $|\res|$ times.\\
     Similarly, the second possible jump from $0$ to $p\left(\frac{2}{3}\right)$ removes one $m$-monochromatic resource and adds one $m-1$-monochromatic resource for $m\in\{2,3\}$ or removes a $1$-monochromatic resource and adds one empty resource. As no agent joins empty or monochromatic resources according to \autoref{lma:IB_monochromatic_no_joins}, this jump type can happen at most $3\cdot|\res|$ times. Given that all possible types of jumps occur only a finite number of times, the IB-FIP holds.
\end{proof}

In \autoref{lma:bin_tree_IB_FIP}, we observe that it is sufficient to use a (lexicographic) potential argument based on the number of empty and monochromatic resources, with monochromatic resources being counted separately depending on the number of agents they contain. However, as $\Delta_G(\res)$ increases, simply counting monochromatic and empty resources becomes insufficient. To address this, we extend the initial approach by an argument based on jumps increasing social welfare, similar to the method used in Lemma 1 by \textcite{harder_strategic_2023}. A limitation arising from this addition is that the $p$-function must be linear for the proof to proceed similarly, thereby slightly weakening the generality of our statement. As a reminder, a $p$-function is called linear if and only if there exists $m_p\in\mathbb{R}_+$ such that for all $x\in(0,\Lambda]$ it holds that $p(x)=m_p\cdot x$. This leads us to the following definition of the potential function $\Phi(\strat)$.

\begin{restatable}[$\Phi$]{definition}{}
    Let $G=(\act\cup\res,\edg)$ be any graph and $\strat$ be any strategy profile. Then $\Phi(\strat)$ is the vector, which contains the number of empty resources in $G$, the number of $1,2,…,\Delta_G(\res)$-monochromatic resources, the number of resources with globally optimal utility for one color and the social welfare in that order. Formally, we define \\
    \[ \Phi(\strat)=
    \left(
    \begin{array}{l}
    |\{q \in \res \mid \text{$q$ is empty in $\strat$}\}|, \\
    |\{q \in \res \mid \text{$q$ is 1-monochromatic in $\strat$}\}|, \\
    |\{q \in \res \mid \text{$q$ is 2-monochromatic in $\strat$}\}|, \\
    \vdots \\
    |\{q \in \res \mid \text{$q$ is $\Delta_G(\res)$-monochromatic in $\strat$}\}|, \\
    |\{q \in \res \mid \text{$\#_R(q,\strat) = \Delta_G(\res) - 1$ or $\#_B(q,\strat) = \Delta_G(\res) - 1$}\}|, \\
    \sum_{a \in \act} u(a, \strat)
    \end{array}
    \right).
    \]
\end{restatable}

\begin{restatable}[]{theorem}{theorem6}
\label{thm:IB_monotonical}
    {For all graphs $G$, a linear $p$-function and $\Lambda \geq \frac{\Delta_G(\res)(\Delta_G(\res)-2)}{\Delta_G(\res)^2-\Delta_G(\res)-1}$ an impact-blind improving move lexicographically increases $\Phi(\strat)$. The number of steps increasing $\Phi$ is limited.}
\end{restatable}

\begin{proof}
    Let w.l.o.g., a red agent $a$ make an impact-blind improving move from resource $q$ to $q'$, changing the strategy profile from $\bold{s}$ to $\bold{s'}$. Let $r_1 = \# _R(q,\bold{s})$, $b_1 = \#_B(q,\bold{s})$, $r_2 = \#_R(q',\bold{s})$, $b_2 = \#_B(q',\bold{s})$. \\
    \textbf{Case 1:} ($r_1 = 1$ and $b_1 = 0$): The jump causes $q$ to become empty in $\strat'$, thus increasing the number of empty resources by one. Note that $q'$ cannot be empty in $\strat$ as the jump would not be impact-blind improving in that case. Hence, as the first element increases from $\Phi(\strat)$ to $\Phi(\strat')$ the potential function $\Phi$ increases lexicographically from $\strat$ to $\strat'$.\\
    \textbf{Case 2:} ($r_1 = 1$ and $b_1 > 0$): The move results in $q$ becoming $b_1$-monochromatic in $\strat'$. According to \autoref{lma:IB_monochromatic_no_joins} $q'$ is not monochromatic in $\strat$ and since $q'$ cannot be empty either for the move to be improving, the number of $b_1$-monochromatic resources increases by one, while the number of empty and $i$-monochromatic resources with $i\in[1,b_1)$ remains unchanged. Consequently, $\Phi$ increases lexicographically from $\strat$ to $\strat'$.\\
    \textbf{Case 3:} ($r_1>1$ and $b_1=0$): In this case, $q$ changes from $r_1$- to $r_1-1$-monochromatic as the strategy shifts from $\strat$ to $\strat'$. Additionally, similarly to Case 2, $q'$ must not have been empty or monochromatic in $\strat$. Therefore, the number of empty resources remains unchanged, while the number of $r_1$-monochromatic resources decreases by one, and the number of $r_1-1$-monochromatic resources increases by one. Since the latter appears earlier in the lexicographic order of $\Phi$, we conclude that $\Phi$ increases lexicographically from $\strat$ to $\strat'$. \\
    \textbf{Case 4:} ($r_1 > 1$, $b_1>0$ and $\frac{r_2+1}{r_2+1+b_2}=\frac{\Delta_G(\res)-1}{\Delta_G(\res)}$): Analogous to the previous cases, $q'$ cannot be empty or monochromatic in $\strat$ for the move to be improving. Since $q'$ gains an agent in the transition to $\strat'$ it holds that $q'$ cannot be empty or monochromatic in $\strat'$ either. Additionally, given $r_1 > 1$ and $b_1>0$, we have that $q$ is not monochromatic or empty in both $s$ and $s'$. Thus, the number of monochromatic and empty resources remains unchanged from $\strat$ to $\strat'$. Since $\frac{r_2+1}{r_2+1+b_2}=\frac{\Delta_G(\res)-1}{\Delta_G(\res)}$, we see that, according to \autoref{lma:largest_fraction}, $q'$ becomes a resource with maximum utility for red agents in $\strat'$. Regarding $q$, note that $q$ cannot provide optimal utility for any color in $\strat$. From $r_1 > 1$, it follows that, without loss of generality, blue agents do not have optimal utility according to \autoref{lma:largest_fraction}. Additionally, if red agents had optimal utility, this would imply that $a$ jumped away from a resource with optimal utility, which leads to a contradiction. Furthermore, as $q$ loses an agent from $\strat$ to $\strat'$, and given that optimal utility requires all adjacent agents of a resource to select that resource, which follows from \autoref{lma:largest_fraction}, it follows that $q$ does not provide optimal utility in $\strat'$ either. Thus, the number of resources providing maximum utility to one color increases by $1$ from $\strat$ to $\strat'$. Consequently, $\Phi$ increases lexicographically from $\strat$ to $\strat'$. \\
    \textbf{Case 5:} ($r_1 > 1$, $b_1>0$ and $\frac{r_2+1}{r_2+1+b_2}<\frac{\Delta_G(\res)-1}{\Delta_G(\res)}$): We begin by noting that, analogous to Case 4, neither $q$ nor $q'$ can be empty or monochromatic in either $\strat$ or $\strat'$. Additionally, since $\frac{r_2+1}{r_2+1+b_2}<\frac{\Delta_G(\res)-1}{\Delta_G(\res)}$ it follows that $q'$ does not yield maximum utility for any color in both $\strat$ and $\strat'$. Thus, the only difference between $\Phi(\strat)$ and $\Phi(\strat')$ lies in the social welfare. It remains to show, that the social welfare always increases with this type of jump. \\
    We observe that since neither $q$ nor $q'$ provides maximum utility in either $\strat$ or $\strat'$, it holds that 
    \begin{equation*}
    \rho_R(q,\strat),\rho_R(q,\strat'),\rho_B(q,\strat),\rho_B(q,\strat')\in(0,\Lambda)
    \end{equation*}
    and
    \begin{equation*}
        \rho_R(q',\strat),\rho_R(q',\strat'),\rho_B(q',\strat),\rho_B(q',\strat')\in(0,\Lambda).
    \end{equation*}
    This follows from utilities increasing with increasing fractions as stated in \autoref{lma:lower_bound_mon_incr_proof}, which implies that if there were any fraction less than $\frac{\Delta_G(\res)-1}{\Delta_G(\res)}$ larger than $\Lambda$, the fraction $\frac{\Delta_G(\res)-1}{\Delta_G(\res)}$ would not be maximal due to $p$ being monotonically decreasing in $[\Lambda,1)$, which contradicts \autoref{lma:largest_fraction}. From this point onwards, our argumentation is analogous to that of \textcite{harder_strategic_2023} in Lemma 1.

    The social welfare at $q$ in $\strat$ is given by 
    \begin{equation*}
        W_q(\strat)
        =m_p\left(\frac{r_1^2+b_1^2}{r_1+b_1}\right)
        =m_p\left(\frac{2b_1^2+r_1^2+r_1b_1-b_1^2-r_1b_1}{r_1+b_1}\right)
        =m_p\left(r_1-b_1+2\left(\frac{b_1^2}{r_1+b_1}\right)\right),
    \end{equation*}
    where $m_p$ is the slope of the linear function $p$ in $[0,\Lambda]$.
    Similarly, the social welfare at $q$ in $\strat'$ is given by
    \begin{align*}
        W_q(\strat')
        &=m_p\left(\frac{(r_1-1)^2+b_1^2}{r_1+b_1-1}\right) \\
        &=m_p\left(\frac{r_1^2+r_1b_1-r_1-r_1b_1-b_1^2+b_1-r_1-b_1+1+2b_1^2}{r_1+b_1-1}\right) \\
        &=m_p\left(r_1-b_1-1+2\left(\frac{b_1^2}{r_1+b_1-1}\right)\right).
    \end{align*}
    Thus, the difference of social welfare at $q$ is
    \begin{align*}
        W_q(\strat')-W_q(\strat)
        =m_p\left(2\left(\frac{b_1^2}{r_1+b_1-1}-\frac{b_1^2}{r_1+b_1}\right)-1\right)
        =m_p\left(\frac{2b_1^2}{(r_1+b_1)(r_1+b_1-1)}-1\right)
    \end{align*}
    Analogously, for $q'$ in $\strat$ we obtain
    \begin{align*}
        W_{q'}(\strat)
        =m_p\left(\frac{r_2^2+b_2^2}{r_2+b_2}\right)
        =m_p\left(\frac{r_2^2+r_2b_2-r_2b_2-b_2^2+2b_2^2}{r_2+b_2}\right)
        =m_p\left(r_2-b_2+2\left(\frac{b_2^2}{r_2+b_2}\right)\right)
    \end{align*}
    and similarly for $\strat'$ we get
    \begin{align*}
        W_{q'}(\strat')
        &=m_p\left(\frac{(r_2+1)^2+b_2^2}{r_2+b_2+1}\right) \\
        &=m_p\left(\frac{r_2^2+r_2b_2+r_2-b_2^2-r_2b_2-b_2+r_2+b_2+1+2b_2}{r_2+b_2+1}\right) \\
        &=m_p\left(r_2-b_2+1+2\left(\frac{b_2^2}{r_2+b_2+1}\right)\right).
    \end{align*}
    For the difference in social welfare at $q'$ we now have
    \begin{align*}
        W_{q'}(\strat')-W_{q'}(\strat)
        =m_p\left(2\left(\frac{b_2^2}{r_2+b_2+1}-\frac{b_2^2}{r_2+b_2}\right)+1\right)
        =m_p\left(1-\frac{2b_2^2}{(r_2+b_2)(r_2+b_2+1)}\right).
    \end{align*}
    Thus, for the difference in social welfare we get
    \begin{align*}
        W(\strat')-W(\strat)
        &=W_q(\strat')-W_q(\strat)+W_{q'}(\strat')-W_{q'}(\strat) \\
        &=2m_p\left(\frac{b_1^2}{(r_1+b_1)(r_1+b_1-1)}-\frac{b_2^2}{(r_2+b_2)(r_2+b_2+1)}\right) \\
        &>2m_p\left(\frac{b_1^2}{(r_1+b_1)^2}-\frac{b_2^2}{(r_2+b_2)^2}\right).
    \end{align*}
    As the move is impact-blind improving we have $\frac{r_1}{r_1+b_1}<\frac{r_2}{r_2+b_2}$, which is equivalent to $\frac{b_1}{r_1+b_1}>\frac{b_2}{r_2+b_2}$. It follows since $m_p>0$ that the social welfare increases from $\strat$ to $\strat'$. \\
    Since nothing aside from the social welfare in $\Phi$ changes from $\strat$ to $\strat'$ the potential function $\Phi$ increases lexicographically. \\
    To see that the number of steps increasing $\Phi$ in a sequence of improving moves is limited, we first observe that the number of empty resources, monochromatic resources and resources, which provide one color with maximum utility can only take integer values in $[0,n]$. Since we additionally obtained a lower bound for the increase in social welfare in Case 5, $\Phi$ can only increase finitely often.
\end{proof}
We now obtain our main result in this section, which follows directly from the previous theorem.
\begin{restatable}[]{theorem}{}
\label{thm:IB-FIP}
    For all graphs $G$, a linear $p$-function and $\Lambda \geq \frac{\Delta_G(\res)(\Delta_G(\res)-2)}{\Delta_G(\res)^2-\Delta_G(\res)-1}$ the IB-FIP holds.
\end{restatable}
\begin{proof}
    According to \autoref{thm:IB_monotonical}, it holds that $\Phi$ increases with every impact-blind improving move and that the number of moves increasing $\Phi$ is limited. Thus, starting from any strategy profile, the number of moves until an equilibrium is reached is limited and the IB-FIP holds.
\end{proof}
To demonstrate that this result is tight, we need to establish that for all $\Lambda$ values outside the interval specified in \autoref{thm:IB-FIP}, the IB-FIP does not hold. We achieve this by showing the existence of instances where no IBE can be found, thereby providing an even stronger result.
\begin{restatable}[]{theorem}{theorem_no_IBE}
\label{thm:no_IBE}
    For all $d\in \mathbb{N}_{>2}$ and $\Lambda < \frac{d(d-2)}{d^2-d-1}$ there exists a graph $G$ with $\Delta_G(\res)=d$ such that no IBE exists.
\end{restatable}

\begin{proof}
    The proof is analogous to that in \autoref{lma:bin_tree_no_IBE}. From the assumption for $\Lambda$ and \autoref{lma:lambda_lower_bound_monotonical} it follows that $p\left(\frac{d-1}{d}\right)<p\left(\frac{d-2}{d-1}\right)$. We now construct a graph similar to the one in \autoref{fig:IBE_counterexample_bin_tree}. We have $\res=\{q_1,q_2\}$, a red agent $a_1$ who has access to both $q_1$ and $q_2$ as well $2$ groups of one blue and $d-2$ red agents who have access only to $q_1$ or $q_2$ respectively.
    We note that both $q_1$ and $q_2$ have degree $d$. We see that $a_1$ has utility $p\left(\frac{d-1}{d}\right)$ independent of their strategy being $q_1$ or $q_2$. Furthermore, they always see utility $p\left(\frac{d-2}{d-1}\right)$ at the resource, at which they are not currently at. Thus, $a_1$ changes their strategy every round and none of the two possible strategy profiles is in IBE. \\
\end{proof}
Finally, this provides us with tight bounds for the existence of equilibria in the game in the impact-blind setting on arbitrary graphs.

\section{Impact-Aware Equilibria}

In this section, we begin by establishing tight bounds on $\Lambda$ such that the IA-FIP holds on cycle graphs and binary trees, with the latter being our main result. For $\Lambda$-values outside the bounds, we strengthen our results by proving the existence of instances which admit no equilibrium. For arbitrary graphs, we narrow down the range of $\Lambda$-values where equilibria may exist to values such that a jump increases utility if and only if it decreases the fraction of same-type agents (except for same-type fraction $0$). Moreover, we illustrate how several conventional methods fail to prove the IA-FIP within the remaining range.

\begin{restatable}[]{lemma}{Lemma2.7}
    On any cycle graph for any $\Lambda \in (0,1)$ the IA-FIP holds.
\end{restatable}

\begin{proof}
    The only two possible utilities for an agent are $0$ and $p\left(\frac{1}{2}\right)$. Since $p\left(\frac{1}{2}\right)$ is optimal, no resource shared by two agents of different colors will ever be left. Therefore, utilities never decrease. Additionally, as there is only one possible increase in utility per agent (from $0$ to $p\left(\frac{1}{2}\right)$), each agent makes at most one move and an IAE is reached in at most $n$ steps.
\end{proof}
Whilst the complexity of the game on cycle graphs is fairly limited, the same cannot be said for binary trees. We divide our proof for the IA-FIP on binary trees for $\Lambda\in\left(0,\frac{1}{2}\right]$ into two cases and use a potential argument for the first interval and an inductive argument for the second one.
\begin{restatable}[]{theorem}{} 
    \label{thm:bin_tree_IA_FIP_Lambda_less_0,4}
    On any binary tree for $\Lambda \in \left(0,\frac{2}{5}\right)$ the IA-FIP holds.
\end{restatable}

\begin{proof}
    From $\Lambda < \frac{2}{5}$ it follows from \autoref{lma:upper_bound_mon_decr_proof}, that $p\left(\frac{1}{3}\right) \geq p\left(\frac{1}{2}\right) > p\left(\frac{2}{3}\right)$. We start by considering the possible jumps in this setting. Whilst doing so, we observe the difference in the number of 1-,2- and 3-monochromatic resources. Note that in \autoref{tab:IA_bin_tree_leq_0,4_states} it is w.l.o.g. always a red agent, who jumps.
    \begin{table}[]
        \centering
        \begin{tabular}{||c c c c c||} 
        \hline
        From & To & 1-monochrom. & 2-monochrom. & 3-monochrom. \\ [0.5ex] 
        \hline\hline
        $\color{red} 1 \color{blue} 0$&$ \color{red} 0 \color{blue} 2$ & -1 & -1 & /\\
        $\color{red} 2 \color{blue} 0$&$ \color{red} 0 \color{blue} 2$ & +1 & -2 & /\\
        $\color{red} 3 \color{blue} 0$&$ \color{red} 0 \color{blue} 2$ & / & -1 & -1\\
        \hline
        $\color{red} 1 \color{blue} 0$&$ \color{red} 0 \color{blue} 1$ & -2 & / & /\\
        $\color{red} 2 \color{blue} 0$&$ \color{red} 0 \color{blue} 1$ & / & -1 & /\\
        $\color{red} 3 \color{blue} 0$&$ \color{red} 0 \color{blue} 1$ & -1 & +1 & -1\\
        \hline
        $\color{red} 1 \color{blue} 0$&$ \color{red} 1 \color{blue} 1$ & -1 & / & /\\
        $\color{red} 2 \color{blue} 0$&$ \color{red} 1 \color{blue} 1 $& +1 & -1 & /\\
        $\color{red} 3 \color{blue} 0$&$ \color{red} 1 \color{blue} 1 $& / & +1 & -1\\
        \hline
        $\color{red} 1 \color{blue} 1$&$ \color{red} 0 \color{blue} 2 $& +1 & -1 & /\\
        \hline
        $\color{red} 2 \color{blue} 1$&$ \color{red} 0 \color{blue} 1 $& -1 & / & /\\
        $\color{red} 2 \color{blue} 1$&$ \color{red} 0 \color{blue} 2 $& / & -1 & /\\
        \hline
        \end{tabular}
        \caption{State transitions of resources in a binary tree with $\Lambda<\frac{2}{5}$.}
        \label{tab:IA_bin_tree_leq_0,4_states}
    \end{table}

    First, we note that there are no transitions, which increase the sum of 1-,2- and 3-monochromatic resources. Furthermore, the only transitions which do not decrease this sum either decrease the amount of 3-monochromatic resources by one while increasing the number of 2-monochromatic resources by one, or decrease the amount of 2-monochromatic resources by one for an increase in 1-monochromatic resources by one. Since the initial amount of 3-monochromatic resources is finite and cannot increase, the first transition-type mentioned can only occur finitely many times. Moreover, this transition-type is the only one that increases 2-monochromatic resources, meaning that the amount of 2-monochromatic resources increases only a finite number of times too. Consequently, transitions, which decrease the number of 2-monochromatic resources by one for an increase in 1-monochromatic resources by one, occur only finitely many times. As the jumps which do not decrease the sum of 1-, 2- or 3-monochromatic resources are limited and the number of jumps decreasing the sum is limited as the number of 1-, 2- and 3-monochromatic resources can only take integer values in $[0,k]$ the IA-FIP holds.
\end{proof}

Before proving the IA-FIP on binary trees for $\Lambda \in \left[\frac{2}{5},\frac{1}{2}\right]$ as well, we first visualize the possible state transitions of resources in such a game and make several key observations, as the potential argument used previously is not suitable for this case. From the value of $\Lambda$, it follows from \autoref{lma:lambda_lower_bound_monotonical} that $p\left(\frac{1}{2}\right)>p\left(\frac{1}{3}\right)\geq p\left(\frac{2}{3}\right)$. Consequently, the possible state transitions are depicted in \autoref{fig:IA_states}.

\begin{figure}
    \centering
    \makebox[0.5\textwidth][c]{\includegraphics{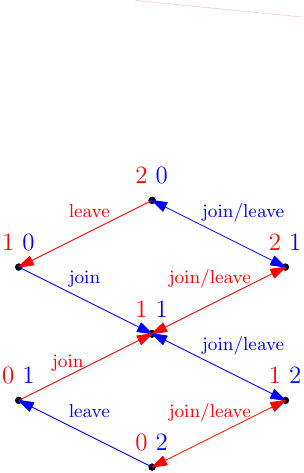}}
    \caption{Overview of transitions between resource states with \textcolor{red}{$3$} \textcolor{blue}{$0$}, \textcolor{red}{$0$} \textcolor{blue}{$3$} and \textcolor{red}{$0$} \textcolor{blue}{$0$} omitted, as the first two states can never be entered and the last one cannot be left.}
    \label{fig:IA_states}
\end{figure}

Our argument is now based on the fact, that agents jump finitely often and resources change their state only a finite number of times when certain conditions are met. As we will see later for agents this is rather simple. For resources, the following three lemmata provide us with the necessary observations.
\begin{restatable}[]{lemma}{Lemma2.10}
\label{lma:outer_loops}
The cycles $\color{red} 2 \color{blue} 1 \leftrightarrow \color{red} 2 \color{blue} 0$ \color{black} and $\color{red} 1 \color{blue} 2 \color{red} \leftrightarrow 0 \color{blue} 2$ \color{black} in the graph in \autoref{fig:IA_states} are followed finitely often by all resources.
\end{restatable}
\begin{proof}
    We consider the cycle $\color{red} 2 \color{blue} 1 \leftrightarrow \color{red} 2 \color{blue} 0$ \color{black} w.l.o.g. as the graph is symmetric for the two colors. \\
    We see that for the join $\color{red} 2 \color{blue} 0 \rightarrow \color{red} 2 \color{blue} 1$ \color{black} a leave from another resource in the form of either $\color{red} 1 \color{blue} 2 \rightarrow \color{red} 1 \color{blue} 1$ \color{black} or $\color{red} 0 \color{blue} 2 \rightarrow \color{red} 0 \color{blue} 1$ \color{black} is necessary. In both cases, a blue agent leaves a resource, which is in the lower cycle. We additionally note, that blue agents will never jump from a resource in the upper cycle to one in the lower cycle, as $p\left(\frac{1}{3}\right)>p\left(\frac{2}{3}\right)$.\\
    Thus, there can be only finitely many joins of blue agents into resources in the state $\color{red} 2 \color{blue} 0$ \color{black} and the cycle is followed at most $|\bold{B}|-1$ times and thus finitely often. Similarly, the loop between $\color{red} 0 \color{blue} 2$ \color{black} and $\color{red} 1 \color{blue} 2$ \color{black} is repeated at most  $|\bold{R}|-1$ times.
\end{proof}

\begin{restatable}[]{lemma}{Lemma2.10}
\label{lma:inner_loops}
The cycles $\color{red} 1 \color{blue} 1 \color{red} \leftrightarrow 2 \color{blue} 1$ \color{black} and $\color{red} 1 \color{blue} 1 \leftrightarrow \color{red} 1 \color{blue} 2$ \color{black} in the graph in \autoref{fig:IA_states} are followed finitely often by all resources.
\end{restatable}
\begin{proof}
    We consider the cycle $\color{red} 1 \color{blue} 1 \color{red} \leftrightarrow 2 \color{blue} 1$ \color{black} w.l.o.g. as the graph is symmetric for the two colors. \\
    We see that the join $\color{red} 1 \color{blue} 1 \color{red} \rightarrow 2 \color{blue} 1$ \color{black} can only happen, if, at an adjacent resource, the leave $\color{red} 2 \color{blue} 0 \color{red} \rightarrow 1 \color{blue} 0$ \color{black} takes place. Meanwhile, the leave $\color{red} 2 \color{blue} 1 \color{red} \rightarrow 1 \color{blue} 1$ \color{black} only happens, when the agent joins an adjacent resource, such that either the transition $\color{red} 0 \color{blue} 1 \color{red} \rightarrow 1 \color{blue} 1$ \color{black} or $\color{red} 0 \color{blue} 2 \color{red} \rightarrow 1 \color{blue} 2$ \color{black} happens, which are both in the lower cycle. We observe, that, once at a resource in the lower cycle, a red agent will never jump back to a resource in the upper cycle, as $p\left(\frac{1}{3}\right)>p\left(\frac{2}{3}\right)$. Thus, with every red agent leaving from the state $\color{red} 2 \color{blue} 1$\color{black}, the number of red agents, which are in the lower cycle permanently decreases by one, and it can never increase again. From this it follows, that the cycle can be repeated at most $|\bold{R}|-1$ times and thus finitely often. Similarly, the cycle $\color{red} 1 \color{blue} 1 \leftrightarrow \color{red} 1 \color{blue} 2$ \color{black} repeated at most $|\bold{B}|-1$ times.
\end{proof}

\autoref{lma:outer_loops} and \autoref{lma:inner_loops} now allow us to formulate the following Lemma, which will serve us in the induction step later.
\begin{restatable}[]{lemma}{Lemma2.12}
    \label{lma:resource_induction_step}
    For $\Lambda \in \left[\frac{2}{5},\frac{1}{2}\right]$, if in a best response sequence either all adjacent red or all adjacent blue agents of a resource $q$ change their strategy only finitely often, then $q$ changes its state only finitely often.
\end{restatable}
\begin{proof}
    Let $q \in \res$ be a resource with either no red or no blue agents adjacent to $q$, which move infinitely often. Assume, that $q$'s state changed infinitely often. This implies, that $q$ changes its state infinitely often caused by the strategy changes of adjacent agents of only one color. Looking at \autoref{fig:IA_states}, we observe that the only cycles regarding $q$'s state, which include only transitions caused by the jump of an agent of one color, are $\color{red} 1 \color{blue} 1 \color{red} \leftrightarrow 2 \color{blue} 1$\color{black}, $\color{red} 1 \color{blue} 1 \rightarrow \color{red} 1 \color{blue} 2$\color{black}, $\color{red} 2 \color{blue} 1 \leftrightarrow \color{red} 2 \color{blue} 0$ \color{black} and $\color{red} 1 \color{blue} 2 \color{red} \leftrightarrow  0 \color{blue} 2$\color{black}. From \autoref{lma:inner_loops} it follows, that the first two cycles only repeat a finite number of times and the same follows for the latter cycles from \autoref{lma:outer_loops}. This contradicts the assumption, that $q$ changes its state only finitely often and thus the lemma holds.
\end{proof}
Finally, we can now prove the IA-FIP for the remaining $\Lambda$-interval.
\begin{restatable}[]{theorem}{theorem_2.9}
    \label{thm:bin_tree_IA_FIP_Lambda_geq_0,4}
    On any binary tree for all $\Lambda \in \left[\frac{2}{5},\frac{1}{2}\right]$ the IA-FIP holds.
\end{restatable}
\begin{proof}
    We prove that every agent and resource changes only finitely often by induction over the binary tree $G=(\act \cup \res, \edg)$ the game takes place on starting from some strategy profile $\strat$. We start by looking at the tree's leaves. If a leaf is an agent, it cannot change its strategy as it has only one adjacent resource. If it is a resource, it has at most one adjacent agent. That means the resource is adjacent to agents of at most one color. Thus, from \autoref{lma:resource_induction_step} it follows, that the resource can change its state only finitely often. \\
    Let $v \in V$ be any node in $G$ and suppose all agents and resources in $v$'s subtree move finitely often. \\
    \textbf{Case 1:} ($v\in \act$):
    Suppose, that $v$ would change its strategy infinitely often. Since every jump by an agent causes the state of two resources to change, this implies that at least two of $v$'s adjacent resources must change their state infinitely often, too. As there is at most one resource adjacent to $v$, that is not in its subtree, this contradicts the assumption, as at least one resource in $v$'s subtree would need to change its state infinitely often. \\
    \textbf{Case 2:} ($v \in \res$):
    As $v$ has at most one adjacent agent, not in its subtree, there is at most one agent adjacent to $v$, which may move infinitely often. Thus, it follows from \autoref{lma:resource_induction_step}, that $v$'s state changes a finite number of times. \\
    
    By induction, it follows that every agent and resource in the tree change their respective strategy or state only finitely often. Therefore, an IAE is reached from $\strat$ after a finite number of jumps and the IA-FIP holds.
\end{proof}

 We now obtain a coherent theorem, about the range for $\Lambda$, over which the IA-FIP holds on binary trees.
\begin{restatable}[]{theorem}{}
    \label{thm:bin_tree_IA_FIP}
    On any binary tree for all $\Lambda \in \left(0,\frac{1}{2}\right]$ the IA-FIP holds.
\end{restatable}

\begin{proof}
    The theorem follows directly from \autoref{thm:bin_tree_IA_FIP_Lambda_less_0,4} and \autoref{thm:bin_tree_IA_FIP_Lambda_geq_0,4}.
\end{proof}

We continue by showing that the bounds in \autoref{thm:bin_tree_IA_FIP} are tight. \\
\begin{restatable}[]{lemma}{theorem_4}
    \label{counterex_IAE_bin_2}
    For $\Lambda \in \left[\frac{1}{2},1\right)$ there exists a binary tree, such that no IAE exists.
\end{restatable}

\begin{proof}
    We first construct the graph depicted in \autoref{fig:IA_bin_tree_counterexample}.
    \begin{figure}
        \centering
        \makebox[\textwidth][c]{\includegraphics{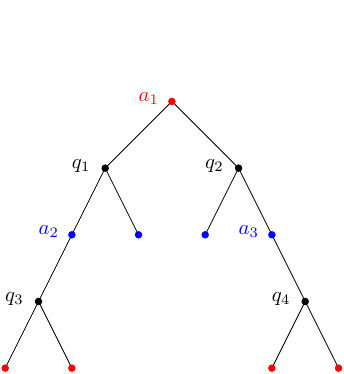}}
        \caption{Red dots are red agents, blue dots are blue agents and black dots are resources. An agent $a$ has access to a resource $q$ if and only if $a$ and $q$ are adjacent in the above graph.}
        \label{fig:IA_bin_tree_counterexample}
    \end{figure} \\
    \textbf{Case 1}: $\Lambda \in \left[\frac{3}{5},1\right)$: It follows from \autoref{lma:lower_bound_mon_incr_proof}, that $p\left(\frac{2}{3}\right)\geq p\left(\frac{1}{2}\right)>p\left(\frac{1}{3}\right)$. For this graph's states, there exists a cycle, which can be seen in \autoref{tab:IA_bin_tree_counterexample}.
    \begin{table}
        \centering
        \begin{tabular}{c c c} 
            \hline
            $a_1$ & $a_2$ & $a_3$ \\ 
            \hline\hline
            $\bold{q_1}, p\left(\frac{1}{2}\right)$ & $\bold{q_3}, p\left(\frac{1}{3}\right)$ & $\bold{q_4}, p\left(\frac{1}{3}\right)$ \\
            \hline
            $\bold{q_1}, p\left(\frac{1}{3}\right)$ & $\bold{q_1}, p\left(\frac{2}{3}\right)$ & $\bold{q_4}, p\left(\frac{1}{3}\right)$ \\
            \hline
            $\bold{q_2}, p\left(\frac{1}{2}\right)$ & $\bold{q_1}, 0$ & $\bold{q_4}, p\left(\frac{1}{3}\right)$ \\
            \hline
            $\bold{q_2}, p\left(\frac{1}{2}\right)$ & $\bold{q_3}, p\left(\frac{1}{3}\right)$ & $\bold{q_4}, p\left(\frac{1}{3}\right)$ \\
            \hline
            $\bold{q_2}, p\left(\frac{1}{3}\right)$ & $\bold{q_3}, p\left(\frac{1}{3}\right)$ & $\bold{q_2}, p\left(\frac{2}{3}\right)$ \\
            \hline
            $\bold{q_1}, p\left(\frac{1}{2}\right)$ & $\bold{q_3}, p\left(\frac{1}{3}\right)$ & $\bold{q_2}, 0$ \\
            \hline\hline
        \end{tabular}
        \caption{An overview over the possible strategy profiles in the game played on the graph in \autoref{fig:IA_bin_tree_counterexample}. Each row is a strategy profile described by the current resource and current utility (in that order) of the agent in the respective column. The strategy profile in every row is reached from the strategy profile above it, except for the top row's strategy profile, which is reached only from the bottom strategy profile.}
        \label{tab:IA_bin_tree_counterexample}
    \end{table}

    The cycle in \autoref{tab:IA_bin_tree_counterexample} contains 6 of the 8 possible strategy profiles. The remaining two strategy profiles are reached, when in the transition from the 3rd to the 4th row of the table, starting from the top, $a_4$ moves before $a_3$ and when in the transition from the 6th to the 1st row of the table $a_3$ moves prior to $a_4$. Note that in the resulting strategy profile $a_3$ and $a_4$ respectively, are the only agents with an improving move, which leads to the states in row five and two of the cycle being reached just the same, regardless of the order. Thus, from every strategy profile, the game goes on indefinitely and no IAE exists. \\

    \textbf{Case 2}: $\Lambda \in [\frac{1}{2},\frac{3}{5})$: It follows from \autoref{lma:lambda_lower_bound_monotonical}, that $p\left(\frac{1}{2}\right)\geq p\left(\frac{2}{3}\right)>p\left(\frac{1}{3}\right)$. We note that compared to case 1, only the relative order of $p\left(\frac{1}{2}\right)$ and $p\left(\frac{2}{3}\right)$ has changed. As in the graph in \autoref{fig:IA_bin_tree_counterexample} $a_1$ can never have utility $p\left(\frac{2}{3}\right)$ as $\mid A(q_1)\cap\red\mid=1=\mid A(q_2)\cap\red\mid$ and both $a_2$ and $a_3$ cannot have utility $p\left(\frac{1}{2}\right)$ as there is always one more blue agent at $q_1$ and $q_2$ and two red agents at $q_3$ and $q_4$ the identical jumps as in case 1 are improving and thus no IAE exists.
\end{proof}

Furthermore, we can show that \autoref{thm:bin_tree_IA_FIP_Lambda_geq_0,4} is tight with regard to $\Delta_G(\res)$ in that the IA-FIP may not hold for $\Lambda \in \left(\frac{\Delta_G(\res)-1}{\Delta_G(\res)^2-\Delta_G(\res)-1},\frac{1}{2}\right]$ for $\Delta_G(\res)>3$.

\begin{restatable}[]{lemma}{Lemma2.3}
     There exists a graph $G=(\act \cup \res,\edg)$ with $\Delta_G(\res)=4$ such that for all $\Lambda \in \left(\frac{3}{11},\frac{1}{2}\right]$ no IAE exists.
     \label{lma:no_IAE_bin_tree_d4}
\end{restatable}

\begin{proof}
    
    As $\Lambda \in \left(\frac{3}{11},\frac{1}{2}\right]$ it follows from \autoref{lma:lambda_lower_bound_monotonical} that $p\left(\frac{1}{4}\right)<p\left(\frac{1}{3}\right)$. Using this property, we now construct the graph in \autoref{fig:counterexample_IAE_below_0,5}. 
    \begin{figure}
        \centering
        \includegraphics[width=0.3\textwidth]{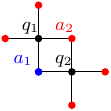}
        \caption{Red dots are red agents, blue dots are blue agents and black dots are resources. An agent $a$ has access to a resource $q$ if and only if $a$ and $q$ are adjacent in the above graph.}
        \label{fig:counterexample_IAE_below_0,5}
    \end{figure}
    For this graph, we observe the cycle of strategy profiles shown in \autoref{tab:IA_graph_counterexample_leq_0,5}.
    
     \begin{table}
        \centering
        \begin{tabular}{c c} 
            \hline
            $a_1$ & $a_2$\\ 
            \hline\hline
            $\bold{q_1}, p\left(\frac{1}{4}\right)$ & $\bold{q_1}, p\left(\frac{3}{4}\right)$\\
            \hline
            $\bold{q_2}, p\left(\frac{1}{3}\right)$ & $\bold{q_1}, 0$\\
            \hline
            $\bold{q_2}, p\left(\frac{1}{4}\right)$ & $\bold{q_1}, p\left(\frac{3}{4}\right)$\\
            \hline
            $\bold{q_1}, p\left(\frac{1}{3}\right)$ & $\bold{q_1}, 0$\\
            \hline\hline
        \end{tabular}
        \caption{An overview over the possible strategy profiles in the game played on the graph in \autoref{fig:counterexample_IAE_below_0,5}. Each row is a strategy profile described by the current resource and current utility (in that order) of the agent in the respective column. note that the strategy profile in every row is reached only from the strategy profile above it, except for the top row's strategy profile, which is reached only from the strategy profile in the bottom row.}
        \label{tab:IA_graph_counterexample_leq_0,5}
    \end{table}
    
    Note that each strategy profile in the table is reached only from the one above it, except for the top strategy profile, which is reached only from the strategy profile in the bottom row. Furthermore, in every strategy profile either $a_1$ or $a_2$, but never both, has an impact-aware improving move, as $a_1$ prefers not sharing a resource with $a_2$ since $p\left(\frac{1}{3}\right)>p\left(\frac{1}{4}\right)$, while $a_2$ seeks to share a resource with $a_1$ as $p(1)<p\left(\frac{3}{4}\right)$. Therefore, for all strategy profiles in the table, there exists only a transition to the following strategy profile in the table. Since there are a total of $2^2$ strategy profiles, all of which are part of the cycle illustrated in \autoref{tab:IA_graph_counterexample_leq_0,5}, no IAE exists.
\end{proof}
After seeing a visual counterexample in \autoref{fig:IA_bin_tree_counterexample} as a part of \autoref{lma:no_IAE_bin_tree_d4} for $\Delta_G(\res)=4$, we can use an analogous argument to see that for all $\Delta_G(\res)>3$ an instance exists such that for all $\Lambda \in \left(\frac{\Delta_G(\res)-1}{\Delta_G(\res)^2-\Delta_G(\res)-1},\frac{1}{2}\right]$ no IAE exists.
\begin{restatable}[]{theorem}{t}
\label{thm:dif_free_arbitrary}
    For all $d>3$ there exists a graph $G=(\act \cup \res,\edg)$ with $\Delta_G(\res)=d$ such that for all $\Lambda \in \left(\frac{d-1}{d^2-d-1},\frac{1}{2}\right]$ no IAE exists.
\end{restatable}

\begin{proof}
    The graph in \autoref{fig:counterexample_IAE_below_0,5} can be altered for any $d>4$ by adding $d-4$ red agents (only) adjacent to $q_1$ and doing the same for $q_2$ to pose a counterexample for $d$. In that case $a_1$ still optimizes for $p\left(\frac{1}{d-1}\right)$, whilst $a_2$ seeks $p\left(\frac{d-1}{d}\right)$ and thus we obtain a generalized cycle, which can be seen in \autoref{tab:IA_graph_counterexample_leq_0,5_generalized}, from which it follows that no IAE exists. \\
    
    \begin{table}
        \centering
        \begin{tabular}{c c} 
            \hline
            $a_1$ & $a_2$\\ 
            \hline\hline
            $\bold{q_1}, p\left(\frac{1}{d}\right)$ & $\bold{q_1}, p\left(\frac{d-1}{d}\right)$\\
            \hline
            $\bold{q_2}, p\left(\frac{1}{d-1}\right)$ & $\bold{q_1}, 0$\\
            \hline
            $\bold{q_2}, p\left(\frac{1}{d}\right)$ & $\bold{q_1}, p\left(\frac{d-1}{d}\right)$\\
            \hline
            $\bold{q_1}, p\left(\frac{1}{d-1}\right)$ & $\bold{q_1}, 0$\\
            \hline\hline
        \end{tabular}
        \caption{An overview over the possible strategy profiles in the game played on a generalization of the graph in \autoref{fig:counterexample_IAE_below_0,5} for arbitrary maximum degrees of a resource in the graph. Each row is a strategy profile described by the current resource and current utility (in that order) of the agent in the respective column. The strategy profile in every row is reached only from the strategy profile above it, except for the top row's strategy profile, which is reached only from the bottom strategy profile.}
        \label{tab:IA_graph_counterexample_leq_0,5_generalized}
    \end{table}
\end{proof}

We conclude this section by considering the case where $\Lambda$ is sufficiently small, such that a decrease in the fraction of same-type agents at a given resource leads to an increase in an agent’s utility (except for reaching a same-type fraction of $0$). Note that, in the following, we consider a fraction in $(0,\frac{1}{2})$ a minority and a fraction in $[\frac{1}{2},1)$ a majority. The following Lemma is similar to Lemma $4$ in the work of \textcite{harder_strategic_2023}.
\begin{restatable}[]{lemma}{} 
    \label{lma:IA_mon_decr_maj_to_min_pot}
    On any graph $G$, if $\Lambda\leq \frac{\Delta_G(\res)-1}{\Delta_G(\res)^2-\Delta_G(\res)-1}$ the potential function $\Phi(\strat)=\sum_{q\in\res} max(\#_R(q,\strat),\#_B(q,\strat))$ never increases. The number of steps decreasing $\Phi(\strat)$ in a sequence of improving moves is limited. 
\end{restatable}

\begin{proof}
    The potential $\Phi$ can only take integer values in $[0,n]$, which limits the number of potential-decreasing moves. Let w.l.o.g. a red agent $a$ make an impact-aware improving move from $q$ to $q'$ changing the strategy profile from $\strat$ to $\strat'$. There are four possible kinds of jumps, between which we will differentiate. \\
    \textbf{Case 1:} ($\#_R(q,\strat)>\#_B(q,\strat)$ and $\#_R(q',\strat)<\#_B(q',\strat)$): We have 
    \begin{equation*}
        \Phi_q(\strat')=max(\#_R(q,\strat'),\#_B(q,\strat'))=max(\#_R(q,\strat),\#_B(q,\strat))-1=\Phi_q(\strat)-1
    \end{equation*} and 
    \begin{equation*}
        \Phi_{q'}(\strat')=max(\#_R(q',\strat'),\#_B(q',\strat'))=max(\#_R(q',\strat),\#_B(q',\strat))=\Phi_{q'}(\strat).
    \end{equation*} Thus, the potential decreases by one. \\
    \textbf{Case 2:} ($\#_R(q,\strat)=\#_B(q,\strat)$ and $\#_R(q',\strat)<\#_B(q',\strat)$): As blue agents stay in the majority both at $q$ and $q'$, the potential stays the same. \\
    \textbf{Case 3:} ($\#_R(q,\strat)>\#_B(q,\strat)$ and $\#_R(q',\strat)\geq\#_B(q',\strat)$): We have 
    \begin{equation*}
        \Phi_q(\strat')=max(\#_R(q,\strat'),\#_B(q,\strat'))=max(\#_R(q,\strat),\#_B(q,\strat))-1=\Phi_q(\strat)-1 
    \end{equation*} and 
    \begin{equation*}
        \Phi_{q'}(\strat')=max(\#_R(q',\strat'),\#_B(q',\strat'))=max(\#_R(q',\strat),\#_B(q',\strat))+1=\Phi_{q'}(\strat)+1
    \end{equation*}. Thus, we have $\Phi(\strat')=\Phi(\strat)-1+1$ and the potential remains unchanged. \\
    \textbf{Case 4:} ($\#_R(q,\strat)<\#_B(q,\strat)$ and $\#_R(q',\strat)<\#_B(q',\strat)$): As blue agents stay in the majority both at $q$ and $q'$, the potential stays the same.
\end{proof}
We can now extract the following result from \autoref{lma:IA_mon_decr_maj_to_min_pot}. Note that 
\begin{restatable}[]{corollary}{}
    On any graph $G$, if $\Lambda\leq \frac{\Delta_G(\res)-1}{\Delta_G(\res)^2-\Delta_G(\res)-1}$ the potential function $\Phi(\strat)=\sum_{q\in\res} max(\#_R(q,\strat),\#_B(q,\strat))$ decreases for a jump from a majority to a minority. The number of steps decreasing $\Phi(\strat)$ in a sequence of improving moves is limited.
\end{restatable}
\begin{proof}
    A jump from a majority to a minority in this setting is covered in case 1 of \autoref{lma:IA_mon_decr_maj_to_min_pot}, in which we see that $\Phi$ decreases by one. Additionally, in  \autoref{lma:IA_mon_decr_maj_to_min_pot} we have seen that the number of potential-decreasing moves is limited. Thus, the statement holds.
\end{proof}

However, while jumps from a majority to a minority are limited, whether other types of jumps can occur infinitely often remains an open question for future research. We finish this section by showing how two conventional lexicographic arguments fail to prove the IA-FIP for these other jump types in this setting.
Concretely we demonstrate, how a lexicographic argument over the minimum or maximum minorities or majorities fail to prove the IA-FIP, when taking both of the remaining jump types into account. \\
We note that for minority to minority jumps, a lexicographic argument over the maximum fractions of minorities would work.
\begin{restatable}[]{lemma}{}
    For $\Lambda\leq \frac{\Delta_G(\res)-1}{\Delta_G(\res)^2-\Delta_G(\res)-1}$ a jump from a minority to a minority decreases the maximum fraction of a minority over the two resources involved in the jump.
\end{restatable}
\begin{proof}
     Let w.l.o.g. a red agent make an impact-aware improving move from $q$ to $q'$, changing the strategy profile from $\strat$ to $\strat'$. Let $r_1 = \# _R(q,\bold{s})$, $b_1 = \#_B(q,\bold{s})$, $r_2 = \#_R(q',\bold{s})$, $b_2 = \#_B(q',\bold{s})$, with $r_1<b_1$ and $r_2<b_2$. Then as $\frac{r_1-1}{r_1+b_1-1}<\frac{r_1}{r_1+b_1}$ and $\frac{r_2+1}{r_2+b_2+1}<\frac{r_1}{r_1+b_1}$, where the latter holds as the move is impact-aware improving, it holds that the minority at $q$ in $\strat$ is larger than the minorities at both resources in $\strat'$.
\end{proof}
However, a jump from a majority to a majority does not necessarily follow this rule.
\begin{restatable}[]{lemma}{}
    For $\Lambda\leq \frac{\Delta_G(\res)-1}{\Delta_G(\res)^2-\Delta_G(\res)-1}$ a jump from a majority to a majority may increase the maximum fraction of a minority over the two resources involved in the jump. 
\end{restatable}
\begin{proof}
    Let w.l.o.g. a red agent make an impact-aware improving move from $q$ to $q'$, changing the strategy profile from $\strat$ to $\strat'$. Let $\# _R(q,\bold{s})=3$, $\#_B(q,\bold{s})=1$, $\#_R(q',\bold{s})=70$, $\#_B(q',\bold{s})=30$. Then blue agents have the fraction $\frac{1}{3}$ at $q$ in $\strat'$, which is greater than both fractions of minorities in $\strat$ at $\frac{1}{4}$ and $\frac{3}{10}$.
\end{proof}
Analogously, for majority to majority jumps, a lexicographic argument over the maximum fractions of majorities would be valid.
\begin{restatable}[]{lemma}{}
    For $\Lambda\leq \frac{\Delta_G(\res)-1}{\Delta_G(\res)^2-\Delta_G(\res)-1}$ a jump from a majority to a majority decreases the maximum fraction of a majority over the two resources involved in the jump.
\end{restatable}
\begin{proof}
     Let w.l.o.g. a red agent make an impact-aware improving move from $q$ to $q'$, changing the strategy profile from $\strat$ to $\strat'$. Let $r_1 = \# _R(q,\bold{s})$, $b_1 = \#_B(q,\bold{s})$, $r_2 = \#_R(q',\bold{s})$, $b_2 = \#_B(q',\bold{s})$, with $r_1>b_1$ and $r_2>b_2$. Then as $\frac{r_1-1}{r_1+b_1-1}<\frac{r_1}{r_1+b_1}$ and $\frac{r_2+1}{r_2+b_2+1}<\frac{r_1}{r_1+b_1}$, where the latter holds as the move is impact-aware improving, it holds that the majority at $q$ in $\strat$ is larger than the majorities at both resources in $\strat'$.
\end{proof}
However, the same does not hold for jumps from a minority to a minority.
\begin{restatable}[]{lemma}{}
    For $\Lambda\leq \frac{\Delta_G(\res)-1}{\Delta_G(\res)^2-\Delta_G(\res)-1}$ a jump from a minority to a minority may increase the maximum fraction of a majority over the two resources involved in the jump.  
\end{restatable}
\begin{proof}
    Let w.l.o.g. a red agent make an impact-aware improving move from $q$ to $q'$, changing the strategy profile from $\strat$ to $\strat'$. Let $\# _R(q,\bold{s})=2$, $\#_B(q,\bold{s})=3$, $\#_R(q',\bold{s})=3$, $\#_B(q',\bold{s})=7$. Then blue agents have the fraction $\frac{3}{4}$ at $q$ in $\strat'$, which is greater than both fractions of majorities in $\strat$ at $\frac{3}{5}$ and $\frac{7}{10}$.
\end{proof}


    \makeatletter
        \def\toclevel@chapter{-1}
        \def\toclevel@section{0}
    \makeatother

    \chapter{Conclusions \& Outlook}
    In this thesis, we initiated the research of the Resource Selection Game with heterogeneous agents in combination with single-peaked utility functions. We differentiated between agents being aware of the impact of them changing their strategy (impact-aware) and agents lacking that knowledge (impact-blind). The latter may be more suited for some scenarios, like families choosing a school for their child, as there is usually a lack of information involved. \\
In the impact-blind setting, our results show that for linear $p$-functions, the game has the IB-FIP if and only if $\Lambda$ is sufficiently large such that an increase in the fraction of an agent's color at their resource causes their utility to increase as well (except for the same-type fraction $1$). How large $\Lambda$ needs to be depends on the maximum degree of a resource in $G$. For all other $\Lambda$, we show that instances exist such that no equilibrium can be reached. This suggests that, in this setting, the IB-FIP holds and the existence of an equilibrium is guaranteed if and only if the agents are quite homophilic. Additionally, for the game played on binary trees, we obtained an analogous result, but without the requirement for the $p$-function to be linear.\\
In the impact-aware setting with the game played on arbitrary graphs we demonstrated that instances for which no IAE can be reached exist for all $\Lambda$ sufficiently large such that an increase in the fraction of an agents color at their resource does not necessarily cause a decrease in utility (excluding jumps away from the fraction $0$, as they always increase utility). For the remaining (smaller) values of $\Lambda$, the question if the game is stable remains open for future research. We provided some observations and ways conventional approaches fail to aid this process. Furthermore, our main result proves tight bounds for the existence of equilibria on binary trees. In \autoref{thm:bin_tree_IA_FIP_Lambda_geq_0,4} we see that, deviating from the impact-blind setting, the properties of the binary tree make it not fall into the generalized statement for arbitrary graphs in \autoref{thm:dif_free_arbitrary}.

    \pagestyle{plain}

    \renewcommand*{\bibfont}{\small}
    \printbibheading
    \addcontentsline{toc}{chapter}{Bibliography}
    \printbibliography[heading = none]

    \addchap{Declaration of Authorship}
    I hereby declare that this thesis is my own unaided work. All direct or indirect sources used are acknowledged as references.\\[6 ex]

\begin{flushleft}
    Potsdam, \today
    \hspace*{2 em}
    \raisebox{-0.9\baselineskip}
    {
        \begin{tabular}{p{5 cm}}
            \hline
            \centering\footnotesize\printAuthor
        \end{tabular}
    }
\end{flushleft}

\end{document}